\documentclass[journal]{IEEEtran}
\usepackage{amsmath,amsfonts}
\pdfoutput=1
\usepackage{textcomp}
\usepackage{algorithmic}
\usepackage{algorithm}   
\usepackage{array}
\ifCLASSOPTIONcompsoc
\usepackage[caption=false,font=normalsize,labelfont=sf,textfont=sf]{subfig}
\else
\usepackage[caption=false,font=footnotesize]{subfig}
\fi
\usepackage{textcomp}
\usepackage{xcolor}
\usepackage{stfloats}
\usepackage{url}
\usepackage{verbatim}
\usepackage{graphicx}
\usepackage{cite}
\usepackage{verbatim} 
\usepackage{mathrsfs, bm, amssymb, bbm}
\usepackage{amsthm}
\usepackage{multirow}

\newtheorem{theorem}{Theorem}
\newtheorem{remark}{Remark}

\begin{document}

\title{High Throughput Inter-Layer Connecting Strategy for Multi-Layer Ultra-Dense Satellite Networks}

\author{\IEEEauthorblockN{Qi~Hao, Di~Zhou,~\IEEEmembership{Member,~IEEE,}
		Min Sheng,~\IEEEmembership{Senior Member,~IEEE}, Yan~Shi, Jiandong Li,~\IEEEmembership{Fellow,~IEEE}}
	%\IEEEauthorblockA{State Key Laboratory of ISN, Information Science Institute, Xidian University, Xi'an, Shaanxi, 710071, China\\
	%Email:  zhoudi@xidian.edu.cn}
	\thanks{Qi Hao, Di Zhou, Min Sheng, Yan Shi, and Jiandong Li are with the State Key Laboratory of Integrated Service Networks, Xidian University, Xi'an, Shaanxi, 710071, China. (email: qhao@stu.xidian.edu.cn; \{zhoudi,yshi\}@xidian.edu.cn; \{msheng,jdli\}@mail.xidian.edu.cn). (Corresponding author: Di Zhou)}
}

% <-this % stops a space
%标注-------
%\thanks{This paper was produced by the IEEE Publication Technology Group. They are in Piscataway, NJ.}% <-this % stops a space
%\thanks{Manuscript received April 19, 2021; revised August 16, 2021.}}



\maketitle

\begin{abstract}
  Multi-layer ultra-dense satellite networks (MLUDSNs) have soared this meteoric to provide vast throughput for globally diverse services. Differing from traditional monolayer constellations, MLUDSNs emphasize the spatial integration among layers, and its throughput may not be simply the sum of throughput of each layer. The hop-count of cross-layer communication paths can be reduced by deploying inter-layer connections (ILCs), augmenting MLUDSN's throughput. Therefore, it remains an open issue how to deploy ILCs to optimize the dynamic MLUDSN topology to dramatically raise throughput gains under multi-layer collaboration. This paper designs an ILC deployment strategy to enhance throughput by revealing the impacts of ILC distribution on reducing hop-count. Since deploying ILCs burdens the satellite with extra communication resource consumption, we model the ILC deployment problem as minimizing the average hop with limited ILCs, to maximize throughput. The proposed problem is a typical integer linear programming (ILP) problem, of which computational complexity is exponential as the satellite scale expands and the time evolves. Based on the symmetrical topology of each layer, we propose a two-phase deployment scheme to halve the problem scale and prioritize stable ILCs to reduce handover-count, which decreases the exponential complexity to a polynomial one, with 1\% estimation error. Simulation results based on realistic mega-constellation information confirm that the optimal number of ILCs is less than $P\cdot S/2$, where $P$ and $S$ are orbits and satellites per orbit. Besides, these ILCs deploy uniformly in each layer, which raises over $1.55\times$ throughput than isolated layers.

\end{abstract}


\section{Introduction}

%第一段先介绍多层超密网络
\IEEEPARstart{W}{ith} the maturity of Low Earth Orbit (LEO) satellite mass production and the low production costs \cite{9466942,9693471,al2022next}, commercial projects have joined the "deep space" adventure, aiming at providing beamless broadband access service for global users \cite{9473799,li2016qos,zhou2023aerospace}. For example, the StarLink constellation designed by the SpaceX company expects to launch 42000 LEO satellites, where 1791 have been in orbits in 10.16.2021 \cite{9385374}. Under these projects, an unprecedentedly complicated large-scale spatial network will be constructed, i.e., the multi-layer ultra-dense
satellite network (MLUDSN). By cooperating with the application of laser or optical inter-satellite links (ISL) and high-frequent user-to-satellite links, the MLUDSN can enhance the service coverage and throughput \cite{al2022next}, and become an appealing architecture of integration of 6G non-terrestrial networks.

Practically, building an effective MLUDSN is aspirational for many plans \cite{Kuiper2019,9211777,9475443}. The explosive network scale entails the whole network overhead sharply. For instance, small satellites with limited bearing usually carry four space communication terminals (SCTs) \cite{2019A}, and extra SCTS that connect with other layers would increase the burden on the satellites. Worse, the increase in scale and the raising layers leaves more satellites underutilized \cite{9829776,zhou2019distributionally}, and causes the end-to-end (E2E) enormous hop-counts\cite{lu2022enhancing}. Therefore, developing inter-satellite connections (ISCs), especially developing inter-layer connections (ILCs), is critical to high-performance MLUDSN design.

\subsection{Inter-Satelltie Connectivity Design}
%从单层的建链开始,
Currently, most of the related works on the ISC development are based on single-layer constellations, such as grid+ topology \cite{2019A,9327501,9461407}, 2D-torus topology \cite{wang2022capacity}, etc. Chen et al. first concluded that optimizing the constellation phase factor could effectively reduce the number of hops through the regularity of the grid+ topology \cite{9351765}. Deng et al. in \cite{9371230,9322362} designed the minimum number of satellite deployments based on the backhaul and global coverage requirements. Authors in \cite{2019A} constructed a new constellation configuration by devising the motif pattern (i.e., each satellite is connected to other satellites with the same number and connection mode), which can reduce the number of ISCs and E2E hop count as much as possible by a fixed one-time configuration. Markedly, it overcomes the disadvantage of grid+ topology in medium- or short-distance communication. However, this structure neglects the uneven feature of global traffic distribution, and each satellite needs to build the same number of ISCs, which is also a huge burden for LEO satellites. To this end, Lin et al. propose a DiV configuration to build a cost-effective constellation by dividing satellites into backbone and access based on the uneven terrestrial traffic \cite{9829776}. The interactive mechanism of ISCs from single layer to multiple layers is still an open issue, particularly, \emph{whether the network capability gain from multi-layer cooperation is a simple superposition of each single-layer capability?}

%另起一段，然而，以上所有的设计与分析是建立在单层网络上，面对多层网络，目前，已有xx考虑xxxx
%However, the aforementioned design and analysis are on the basis of monolayer satellites. Basically, 
The existing ILC development of the multi-layer constellation is still focused on the network built by LEOs and Medium Earth Orbit satellites (MEOs), or LEOs and Geosynchronous Earth Orbit satellites (GEOs), or LEOs and MEOs and GEOs \cite{6692690,6353997,6494333,9685355,8844689,zhou2020research,huang2022multi}. Huang et al. prioritized ILCs by setting multiple weights such as link length and link load, to guarantee  minimal power consumption and minimal handover counts \cite{huang2022multi}. But the ILC design method ignored the impacts of ILCs on the whole network service capability. In this regard, Liu et al. analyzed the capacity of the MEO-LEO network with the ILCs of the shortest distance, and pointed out that the network capacity was approximately equal to the total capacity of the two layers \cite{liu2015capacity}. Derived from the ILC design method, authors concentrated on the efficient load balance between LEOs and MEOs with respect to the transmission delay, traffic migration, and the optimal deployed height of LEOs and MEOs \cite{6692690,6353997,6494333}. Moreover, Wang et al. in \cite{9685355} exploited both stochastic geometry and queueing theory to minimize the total satellite number for fulfilling seamless coverage. Then there are simulation results that evaluated the impacts of ILCs on the latency of data packets in \cite{6338484}. %Besides, \cite{8844689,zhou2020research} adopted contact plan design and snapshot respectively to depict the motion trajectory of spatial nodes and the time evolution nature of the GEO-MEO-LEO network. 

Nevertheless, the aforementioned studies have not sufficiently investigated the network transmission capability e.g. network throughput, on multi-layer collaboration, comprising the impacts of the number of ILCs and the deploying complexity of them. Considering that future satellites may carry five or more SCTs \cite{9393372} based on the Federal Communications Commission (FCC) filing, we explore the transmission gain from inter-layer connectivity under multi-layer collaboration from the following two questions:
\begin{itemize}
  \item \emph{Considering the cost of ILCs, how many ILCs should be built in a given MLUDSN to maximize throughput?}
  \item \emph{How to enhance the stability of ILCs to reduce the handovers when there exists an extreme difference between different time slices in the total topology?}
\end{itemize}

%%% 根据两个问题展开
Intuitively, adding ILCs is equivalent to adding link resources inherently, the E2E routing hops of multiple satellite pairs would be significantly reduced in the whole network. Hence the network transmission capability is dramatically enhanced. However, blindly adding ILCs would bring more design challenges and extra SCT costs to actual industrial production. One of the core questions is how to realize the efficient tradeoff between the performance gain and the cost. 
Besides, since the satellites in different layers would construct the intermittent connections based on the relative movements, the total topology may be changeable over time slices, which imposes the satellite configuration counts and the number of handovers. The unique feature of time-varying topology would increase the computational complexity of constructing the optimal ILC strategy with satellite scale. For example, even in  a snapshot, the complexity of searching for the optimal ILC problem is $\mathcal{O}(A\cdot 2^B)$, where $A$ is a general coefficient. As the network scale (i.e. $B$) increases, its complexity increases exponentially, or even factorially if considering the traffic matrix \cite{Kuiper2019}. Therefore, the second core problem is to minimize the number of handovers of each ILC to maintain the stability of the entire topology.

\subsection{Main Contributions}

In this paper, we explore the intertwined impacts of the ILC distribution on the network throughput. Considering the deploying expenditure of ILCs, we formulate the ILC deployment problem as maximizing throughput with minimizing the number of ILCs. Note that network average hop can be simplified as the proxy of throughput since it depends on the topological structure and traffic routing \cite{Kuiper2019}. 
% investigate the impacts of ILC deployment on network throughput. Specifically, we first formulate the problem of the cost-effective ILC design by maximizing throughput with the lowest costs of deploying ILCs.
%Considering maximizing throughput can be transformed as minimizing the total average path length (APL, i.e., average hop) of the entire network \cite{Kuiper2019}, 
we reformulate the proposed problem as minimizing the product of the number of ILCs and the total APL, to pursue the effective tradeoff between the ILC resource consumption and transmission efficiency. The problem is a typical integer linear programming (ILP) problem and is intractable by exhaustive search on a large network scale. Meanwhile, the problem only considers the link configuration for one snapshot and ignores the interactive relationship of cross-layer satellites between time slots.  Therefore, we propose a two-phase inter-layer connection deployment (TPILCD) algorithm to decouple the problem into two phases, i.e., the candidate ILC set search and maximizing time weight matching, where the time weight is a metric to depict the total duration and remaining duration of ILCs. The proposed algorithm can reduce the computational complexity from exponential one to polynomial, with 1\% error on a small scale. Simulation results demonstrate the performance gains of the proposed algorithm compared with existing benchmarks and guide the inter-layer topology design of the MLUDSN.

\begin{figure}
  \centering
  \includegraphics[scale=0.12]{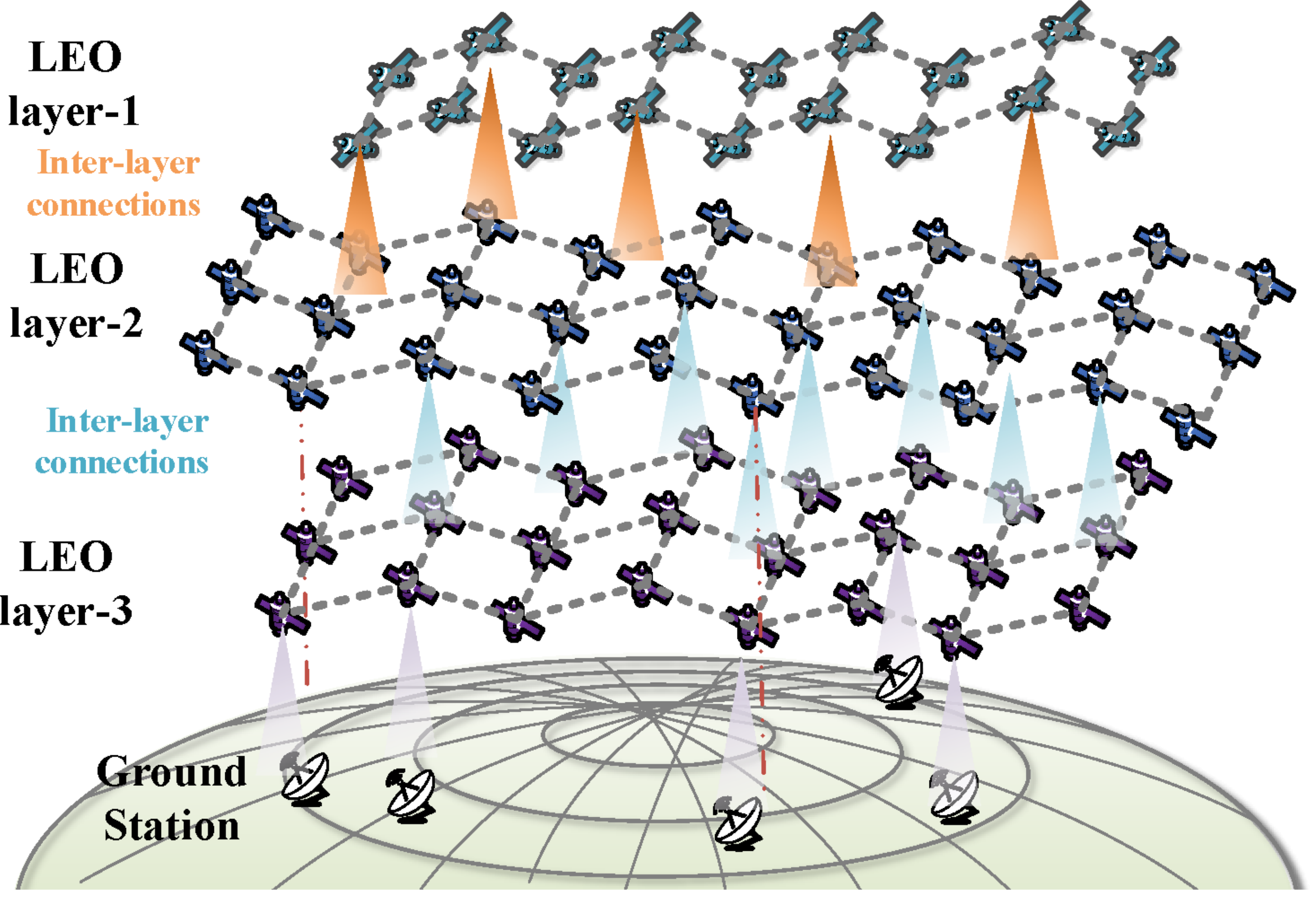}
  \caption{An illustration of a typical MLUDSN.}
  \label{fig1}
\end{figure}

The key contributions are summarized as follows:
\begin{itemize}
   \item The ILC deployment problem is formulated as maximizing throughput with minimizing the number of ILCs, which can be reformulated as minimizing the product of the sum of SCTs and APL by using the inversely proportional relationship of APL and throughput, to build a tractable ILP problem.
   
  \item An analytical mathematical relationship among the APL, network size and the number of ILCs is modeled, to infer that the trend of the APL of the MLUDSN decays logarithmically as the number of ILCs rises.

  \item The two-phase deployment scheme is proposed to halve the problem scale and to select stable ILCs, which can reduce the computational complexity from exponential one to polynomial, with 1\% error on a small scale.
  
  \item Simulation results confirm that the optimal number of ILCs $k$ is $k\leq N_2/2$, when the number of satellites in the two adjacent layers satisfies $N_1>N_2>\ln N_1+3$, while these ILCs deploy evenly in each layer, total throughput can augment over $1.55\times$ than the simple addition of isolated layers.
  
  \item Extensive simulations are conducted on global traffic generated by vast cities with dense populations, to demonstrate that compared to independent layers and other algorithms, cross-layer communication can raise 20.77\% throughput and reduce $2\times$ handover counts and 0.4 hops.
\end{itemize}

The rest of the paper is organized as follows. We present a system model for the MLUDSN including network, ISL model, and multi-layer super-adjacency matrix in Section II. Then Section III formulates the ILC deployment problem with physical constraints and models the APL. Section IV presents our two-phase algorithm and analyzes the computational complexity. The simulation results of algorithm performance and topology design guidance are presented in Section VII. Finally, Section VII concludes the paper.

%system model: include network model, ISL & its rate model, and Graph model with supra-adjacency matrix!
\section{System Model}
%In this section, we briefly introduce the network model, ISL model, and multi-layer super-adjacency matrix to present the construction of the MLUDSN.
In this section, we briefly introduce the network model, ISL model and graph model to depict the basic feature of the MLUDSN.
\subsection{Network Model}
We consider a typical MLUDSN that is composed of $l$ monolayer LEO constellations. Each LEO layer is a general Walker Delta LEO constellation, defined with a 5-tuple $N_i/P_i/F_i:H_i:\epsilon_i$. Here $H_i$ and $\epsilon_i$ represent the altitude and inclination of the $i$-th layer. $N_i = P_i\times S_i,i=1,2,...,l+1$ indicates the number of orbital planes times the number of satellites per orbit in the $i$-th LEO layer. $F_i$ represents the relative phase difference between satellites of adjacent orbits passing through the equatorial plane \cite{9829776}, which is measured by $1/P_i$, and $0\leq F_i \leq S_i-1$. Then we denote $\mathbf{N}=\{N_1,N_2,...N_i,...,N_l\}$, as a set comprising the number of satellites in each layer. Besides, each satellite is equipped with a total of no more than six antennas for intra-layer and inter-layer connectivities. We further define the satellites moving along the increasing latitude direction when the orbital inclination is less than $90^{\circ}$ as ascending satellites. In this way, the satellites moving along the decreasing latitude direction can be regarded as descending satellites.

\subsection{ISL Model}
We classify the inter-satellite links into two types: intra-layer links and inter-layer links. Fig. \ref{fig1} illustrates a widely adopted intra-layer link connecting mode \cite{9322362,9371230}. For example, in a Walker constellation, each satellite establishes four permanent intra-layer links: two intra-plane links and two inter-plane links. Different from stable intra-layer links, inter-layer links exist among layers. Due to the rapid movement of satellites and the constraint of visibility, building ILCs is extremely challenging.

%层间链路是在可见性等约束下构造的可行集的基础上进一步选出的
Intuitively, ILCs are selected from the candidate link set which is constructed by strict constraints such as visibility. In this way, we build the candidate link set with two constraints, i.e. visibility and co-motion feature. The specific descriptions are as follows:

1) Visibility: due to the various inclinations and heights of different layers, ILCs need to consider more physical constraints. As depicted in Fig. \ref{fig_4}, arbitrary two satellites $u$ and $v$, from the $i$-th layer and the $j$-th layer respectively must satisfy a condition, i.e. the distance $d_{uv}$ between them cannot exceed the maximum visual distance $d_{max}$. Firstly, the Euclidean distance $d_{uv}$ is given in spherical coordinates as
%3个表达式：它们之间的距离，最大距离，以及对应的关系
\begin{equation} 
  d_{uv}=\left\{
  \begin{aligned}
      &  (h_i+R_e)^2 + (h_j+R_e)^2-2(h_i+R_e)(h_j+R_e)\\
      &\times(\cos \phi_u \cos \phi_v + \cos (\epsilon_i-\epsilon_j)\sin\phi_u\sin\phi_v)
  \end{aligned}
 \right\}^{\frac{1}{2}},
\end{equation}
where $\phi_u$ and $\phi_v$ are the longitude of satellite $u$ and $v$, respectively. $R_e$ is the earth's radius.
Next, the max line-of-sight distance between them is given as 
\begin{equation}
      d_{max}= \sqrt{h_i(h_i+2R_e)}+\sqrt{h_j(h_j+2R_e)}.
\end{equation}
Thus, the condition can be formulated as 
\begin{equation}
  d_{uv}\leq d_{max}.
  \label{con-1}
\end{equation}

2) Co-motion feature: the direction of motion of the satellite should also be taken into account. When two satellites move toward each other in opposite directions, the link may be very unstable. Therefore, the inter-layer links between satellites moving in the same direction should be established first. Based on the aforementioned ascending and descending satellites, we define the latitude difference between the satellite's adjacent time slots as $\Delta \theta$, thus we have 
\begin{equation}
  \Delta \theta_u \times \Delta \theta_v > 0,
  \label{con-2}
\end{equation}
which indicates that they are both ascending satellites or they are both descending satellites.

%添加ISL的速率计算公式
Besides, we model the link rate of ISLs by Shannon's theorems \cite{9327501}, which are presented as:
\begin{equation}
  \begin{aligned}
  \mathcal{R}_{uv} &= B\cdot \log (1+SNR_{uv})\\
  &= B\cdot \log (1+\frac{P_t\cdot G_{0}^2}{k_B\cdot\tau\cdot B \cdot F_{uv}})
  \end{aligned}
  \label{rate1}
\end{equation} 
where $B$ is the channel bandwidth in Hertz and $SNR_{u,v}$ is the signal-to-noise ratio for an ongoing transmission from the satellite $u$ to $v$ and vice-versa. Then $P_t$ is transmission power, $G_0$ is the maximum gain of transmitter and receiver antennas, $k_B=1.381\cdot 10^{-23}$ is the Boltzmann constant, $\tau=354.18$ is the thermal noise in Kelvin, and $F_{uv}= \bigg(\frac{4\cdot \pi \cdot d_{uv}\cdot f}{c}\bigg)^2$ is the free-space path loss between $u$ and $v$.

%\subsection{Supra-adjacency matrix of MLUDSN}
\subsection{Graph Model}
% 先，再说明超级邻接矩阵
Let us consider a MLUDSN denoted by $\mathcal{G}=(\mathcal{V},\mathcal{E})$ that consist of $l$ layers $\mathcal{G}_1$, $\mathcal{G}_2$,..., $\mathcal{G}_l$, where $\mathcal{V}$ and $\mathcal{E}$ comprises the total satellites and ISLs in the MLUDSN, respectively. Each of these layers is denoted as $\mathcal{G}_i(\mathcal{V}_i,\mathcal{E}_i)$, where $\mathcal{V}_i \in \mathcal{V}$ and $\mathcal{E}_i \in \mathcal{E}$ represent the satellites set and intra-layer links of layer $l$, respectively. Totally, the supra-adjacency matrix of a MLUDSN $\mathcal{G}$ can be defined in a block-matrix structure as \cite{wider2016ensemble} 
\begin{equation}
  \mathcal{A}=
  \begin{pmatrix}
    \mathcal{A}_1 & \cdots & \mathcal{A}_{1,l} \\
    \vdots & \ddots & \vdots  \\
    \mathcal{A}_{l,1} & \cdots & \mathcal{A}_l 
  \end{pmatrix},
\end{equation}
Here $\mathcal{A}_1,..., \mathcal{A}_l$ correspond to the layers $\mathcal{G}_1$, $\mathcal{G}_2$,..., $\mathcal{G}_l$ on the diagonal, thus entries of these block matrices indicate the intra-layer links of the MLUDSN. Off-diagonal matrices $\mathcal{A}_{ij}$ for $i,j \in {1,2,...,l}$ with $i\neq j$ represent ILCs that connect satellites in layer $\mathcal{G}_i$ to satellites in layer $\mathcal{G}_j$. Considering the MLUDSN is undirected network, we also have $\mathcal{A}_{ij}^\top = \mathcal{A}_{ji}$. 

\begin{figure}
  \centering
  \includegraphics[width=2 in]{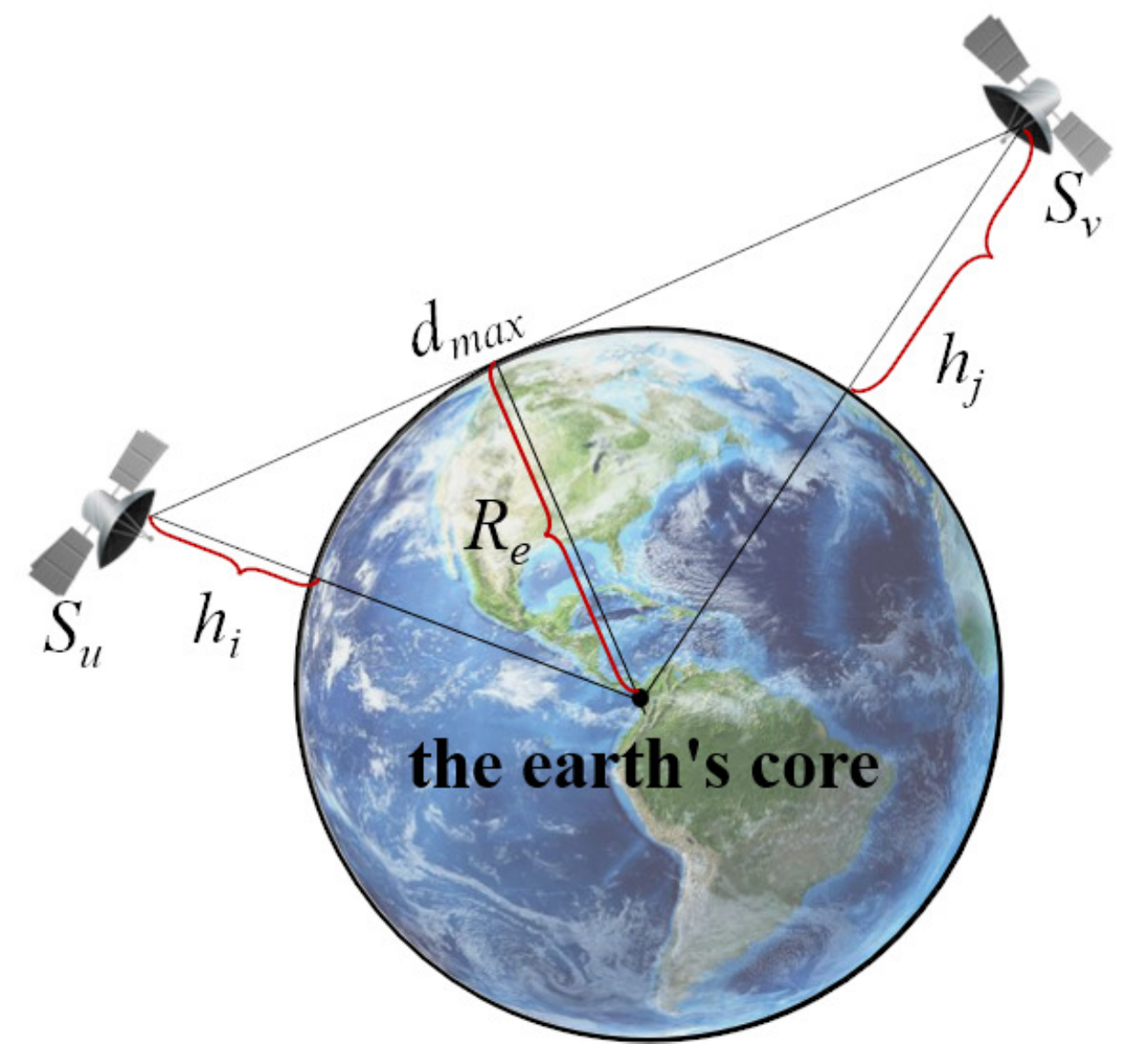}
  \caption{The visibility of two satellites.}
  \label{fig_4}
\end{figure}

\section{ILC Deployment Problem Formulation}
%In this section, we first formulate the ILC deployment problem with the physical constraints. We also analyze the inherent features of ILCs and the symmetry of total topology to simplify the problem. Furthermore, we model the APL of the total network by united degree distribution probability to exploit the relationship between the approximate APL and the number of ILCs, which can efficiently reduce the computational complexity of the proposed problem.

In this section, we first 

\subsection{Cost-Throughput Tradeoff Problem Formulation}
%引出吞吐量的定义，

Throughput is a fundamental property of networks, which is deeply involved with network topology and traffic matrix\footnote{Our focus in this paper is on high-throughput network topology, of which network throughput is end-to-end throughput supported by the MLUDSN in a fluid-flow model with optimal routing. }. We denote $\mathcal{F}$ as a traffic matrix, where $c_{uv}^{ij}$ indicates an amount of requested flow passing through the ISL $e_{ij}$ from $u$ to $v$. The throughput of the MLUDSN with TM $\mathcal{F}$ is the maximum value $\mathcal{Q}$ for which $\mathcal{F}\cdot \mathcal{Q}$ is feasible in the MLUDSN, subject to the flow conservation constraints and the ISL capacity. Therefore, $\mathcal{Q}$ can be obtained by modeling maximum concurrent flow problems \cite{2014throughputPro}. Let $y_{ij}$ indicates whether there is a ISL between $v_i$ and $v_j$, and $\gamma_f$ indicates whether the $f$-th task flow is successfully completed or not. Thus we first formulate the problem as follows:
\begin{equation*}
  \begin{aligned}
      %\mathbb{P}_1: &\quad \quad \max \quad \mathcal{Q}=\sum_{f_{uv}\in \mathcal{F}} \sum_{e_{ij}\in \mathcal{E}}y_{ij}\cdot c_{ij}^{uv} \\
      \mathbb{P}_1: &\quad\quad \max_{y_{ij}\in \mathbf{Y}, c_{uv}^{ij}\in \mathbf{C}} \min_{1\leq f \leq \vert \mathbf{F} \vert } \sum_{f} \gamma_f  \\
      s.t. & \sum_{i=1}^{\vert \mathbf{N} \vert} y_{ij}\cdot c_{ij}^{uv} = \sum_{i=1}^{\vert \mathbf{N} \vert} y_{ji}\cdot c_{ji}^{uv},  \forall f_{uv} \in \mathcal{F}, v_j \in \mathcal{V}, j\neq u,v,  \\
      & \sum_{i=1}^{\vert \mathbf{N} \vert} y_{uj} c_{uj}^{uv}  = \sum_{i=1}^{\vert \mathbf{N} \vert} y_{iv} c_{iv}^{uv} = r_{uv}\cdot \gamma_f, \forall f_{uv}\in \mathcal{F}, \\
      & \sum_{f_{uv}\in \mathcal{F}} y_{ij} \cdot c_{ij}^{uv}\cdot \gamma_f \leq C_{ij}, \forall e_{ij}\in \mathcal{E}, \\
      & \sum_{i=1}^{\vert \mathbf{N} \vert} y_{ij}\leq \alpha_{max}, \forall v_j \in \mathcal{V}. 
  \end{aligned}
\end{equation*}
Our objective is to maximize the minimum throughput of any requested end-to-end flow. The first constraint is the flow conservation for the intermediate relay satellites, and the second constraint is the flow conservation for the origin and the destination of the $f$-th flow, where $r_{uv}$ represents the volume of the flow. The third constraint depicts that the aggregate data of the flow through the ISL should not exceed the maximum amount of task data that can be transmitted by the ISL. The last constraint illustrates that each satellite can build $\alpha_{max}$ links to connect with other satellites.

%说明该问题的难度之大，因此需要转换问题去求解，将问题变为对吞吐量的估计并采用启发式的方式去近似最优的吞吐量。
Obviously, the problem is intractable with the explosive satellite scale, not to mention multi-slots MLUDSN's throughput. Therefore, we design a heuristic model to approximate the optimal throughput. 
%通过上述的分析可知，由于影响吞吐量的因素为流量模型和网络拓扑设计，需要注意的是由于添加层间链后坚持采用all-to-all流量模型会改变流量矩阵本身，进而无法体现网络拓扑改变对吞吐量的影响，因此，这里采用多种流量模型来减弱这个不利因素对层间链性能的影响。

%Throughput performance is one of the key properties in the network \cite{9067004}. 
%A forward-looking yet significant problem that we should study is how to maximize throughput by adding a minimum of extra ILC resources in a typical MLUDSN. 

%% 强调吞吐量的近似解---重写
%The maximum throughput of the network refers to the maximum data rate that a source node can transmit to a destination node when the actual transmission rate of each link is smaller than the maximum feasible transmission rate. It also can be regarded as the Max-flow problem, and be solved by the Max-flow Min-cut method. This method is suitable for all nodes with the same arrival rate in the all-to-all traffic model, but is not suitable for other traffic models. In this regard, we consider a particular approach to obtaining the maximum throughput by the proportion of the sum of link rate and APL in unit time. 

Considering the tandem feature of the MLUDSN, we exemplify two layers to formulate throughput problem. We introduce a binary variable $y_{i,j}$ to indicate whether there exists an ILC between satellite $i$ and $j$, where $i$ and $j$ belong to different layers.
Thus, the objective function is presented as 
\begin{equation}
  \mathbb{Q}_1 = \min \left\lVert  \text{argmax}_{\mathbf{Y}} \frac{\sum_{i\in N_1}\sum_{j\in N_2} y_{i,j}\cdot \mathcal{R}_{i,j}}{\mathcal{P}_{N_1,N_2}} \right\rVert,
\end{equation}
where $\mathbf{Y}$ is the set of $y_{i,j}$,  $\mathcal{R}_{i,j}$ means the transmission rate between $i$ and $j$, and $\mathcal{P}_{N_1,N_2}$ indicates the total average path length of the two-layer MLUDSN, and $N_1$, $N_2$ are the number of satellites in the corresponding layer, respectively.

The objective for us is to maximize the total throughput with the minimum number of ILCs, and the conditions are listed as follows:

1) Each satellite in each layer only has one SCT \cite{9219130} that can connect with the satellite in other layer.

2) The physical constraints: visibility and co-motion feature.

%展示时间权重，减少切换次数， 保持较为稳定的拓扑
\begin{figure}
  \centering
  \subfloat[]{\label{fig_7a}\includegraphics[scale=0.13]{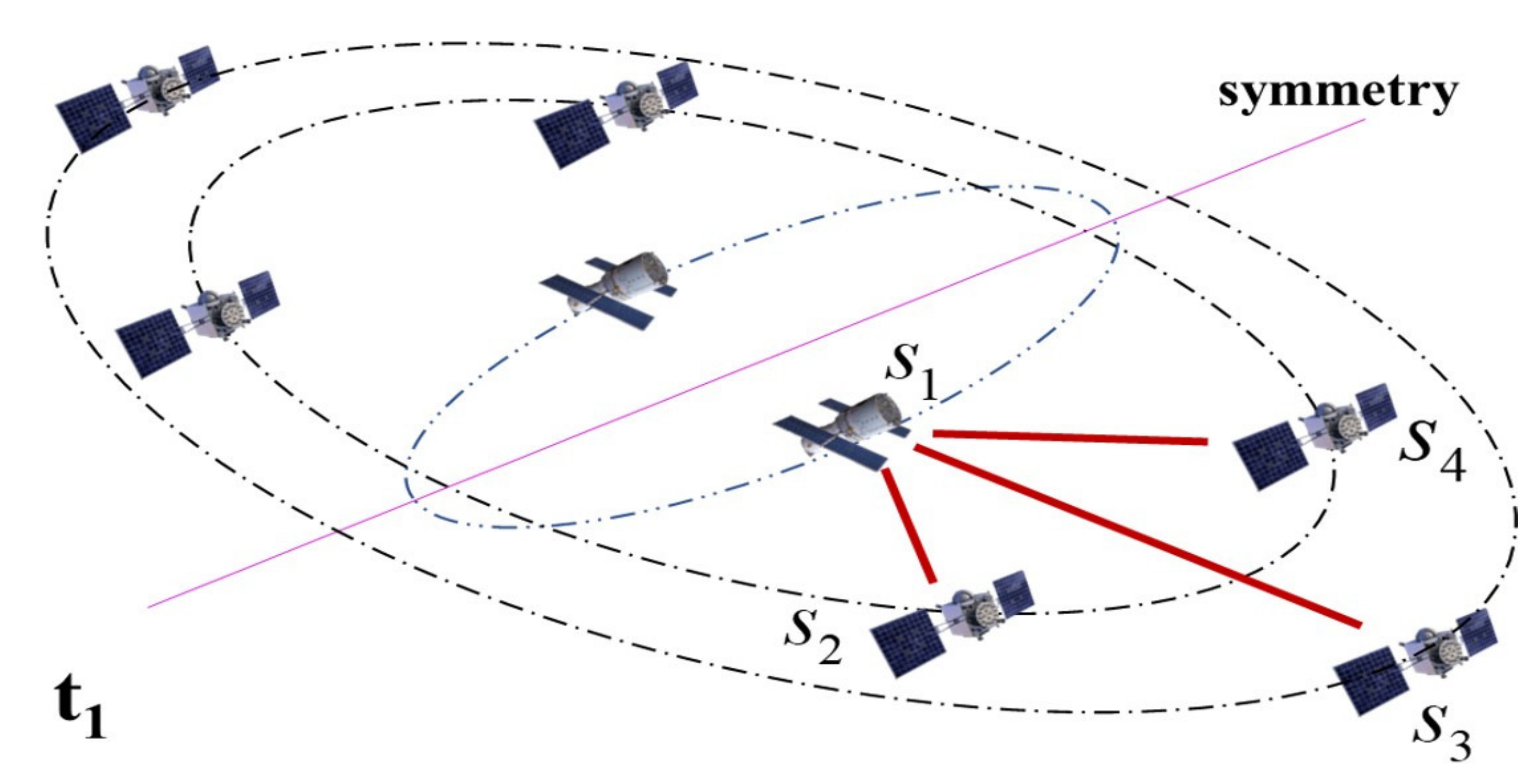}} 
  \\ \vspace{-1em}
  \subfloat[]{\label{fig_7b}\includegraphics[scale=0.18]{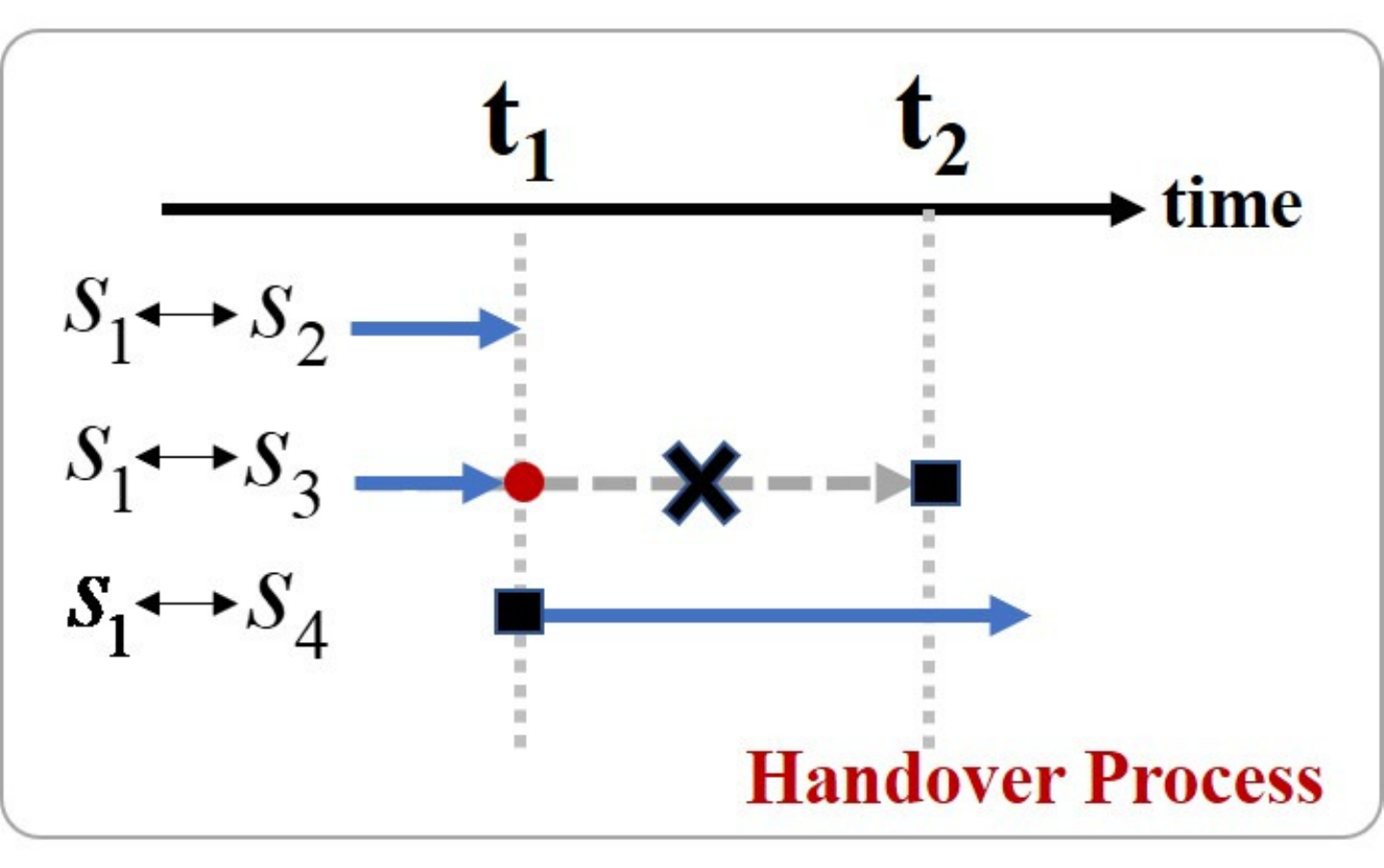}}
  \caption{(a) the symmetry of topology; (b) the handover optimization based on the time metric.}
\end{figure}

Through the aforementioned analysis of throughput, we can observe that throughput is inversely proportional to the total APL. Thus, we regard the APL as the proxy of throughput, which is the product of the number of ILCs and the total APL:
\begin{equation}
  \begin{aligned}
    \mathbb{Q}_2=&\min \sum_{i \in N_1}\sum_{j\in N_2} y_{i,j} \cdot \quad \mathcal{P}_{N_1,N_2}, \\
    s.t.  & \quad 0\leq \sum_{i\in N_1} y_{i,j} \leq 1, \forall j={1,2,\dots,N_2}, \\
    & \quad 0\leq \sum_{j\in N_2} y_{i,j} \leq 1, \forall i={1,2,\dots,N_1}, \\ 
    % & \quad 1 \leq\sum_{j=(p-1)\cdot S_2 +1}^{p_2\cdot S_2} y_{i,j} \leq \lambda_2, \forall i=1,2...N_1, p=1,2,...P_2, \\
    % & \quad 1 \leq\sum_{i=(p-1)\cdot S_1 +1}^{p\cdot S_1} y_{i,j} \leq \lambda_1, \forall j=1,2...N_2, p=1,2,...P_1, \\
    & \quad d_{i,j} \leq d_{max}, \forall i=1,2,...,N_1, j=1,2,...,N_2,\\
    & \quad \Delta\theta_i \times \Delta\theta_j  >0, \forall i=1,2,...,N_1, j=1,2,...,N_2,
  \end{aligned}
\end{equation} 
%the third and fourth constraints correspond to the aforementioned second condition ($\lambda$ indicates the maximal connections of each orbit in the layer); 
where the first constraint and second constraints correspond to the aforementioned first condition,the thrid and fourth constraints correspond to the second condition.

The proposed problem is obviously an ILP problem, which only considers a snapshot, but it is already NP-hard. Specifically, in one snapshot, the complexity of the ILP problem is deeply related to the number of binary variables and the corresponding constraints. Thus, the number of binary variables is $U_1=N_1\cdot N_2 +N_1+N_2$, and the number of constraints is $U_2=N_1+N_2$. Totally, the complexity of this problem is $U_1\cdot U_2 \cdot 2^{U_1}$. With the increase in network scale, its executive time is unacceptable. Besides, acquiring the APL also needs an exhaustive search for the topology after adding ILCs. Motivated by these factors, we first consider how to reduce the number of variables and constraints based on the inherent features of ILCs and the symmetry of topology over the long duration. Then we exploit the change of degree distribution of total topology to derive approximate APL to further reduce the computational complexity. 

\subsection{Time Metric of Inter-layer Connections}
Based on the aforementioned analysis of the problem, the core challenge is how to narrow the feasible region of the proposed problem and pursue a stable topology over a long duration. Therefore, we weigh ILCs in terms of time metric. As shown in Fig. \ref{fig_7a}, the visible matrix of ILCs is almost symmetrical due to the occlusion of the Earth. We consider only half satellites of each orbit in layer $i$ is possible to connect with half of the satellites of each orbit in layer $j$, i.e., the possibly feasible number of satellites that can deploy ILCs are $P_i\cdot S_i/2$ and $P_j\cdot S_j/2$, respectively. 

In Fig. \ref{fig_7b}, we further exploit the stability of total topology to reduce the number of handovers and realize better routing performance. Specifically, we not only consider the total remaining duration of each ILC but also take account of the residual time as the time evolves \cite{huang2022multi}. 
Assuming $S_1$ is possible to connect with $S_2,S_3$ simultaneously before $t_1$. At time slot $t_2$, $S_1$ would switch from $S_2$ to $S_4$ and $S_1$ is invisible to $S_3$. Based on the total remaining duration and residual time, the number of handovers can be decreased. In short, the remaining time normalized weight of each ILC can be formulated as
\begin{equation}
  \mathcal{T}_{y_{i,j}}(t)= \frac{\Delta T_{y_{i,j}}(t)}{\max{\Delta T}}\cdot \frac{\Delta r_{y_{i,j}}(t)}{\max{\Delta r}},
  \label{time_con}
\end{equation}

where $\Delta T_{y_{i,j}}(t)$ means the total remaining time of ILC $y_{i,j}$ at the time $t$, and $\Delta r_{y_{i,j}}(t)$ indicates the residual time of ILC $y_{i,j}$ at the time $t$. According to the weighted links, we can further set the adjustable threshold $\eta_1$ and $\eta_2$ to restrain the feasible region:
\begin{equation}
  \eta_1 \leq \mathcal{T}_{y_{i,j}}(t)\leq \eta_2, \forall y_{i,j}\in \mathbf{Y}, t\in [0,1,...,\tau].
  \label{time} 
\end{equation}

The existence of lower bound $\eta_1$ illustrates that if the remaining time of the ILC is too short, the connection may disappear directly or the number of handovers significantly increases. Meanwhile, the upper bound $\eta_2$ means that if the distance of the ILC is large, the number of handovers may decrease, but the link building imposes power consumption. Totally, last two constraint of the original problem and Eq. \eqref{time} can construct a feasible set in advance, denoted as $\mathbf{F}$, and $\mathbb{Q}_2$ is further solved in this feasible set. Totally, the problem is reformulated as 
\begin{equation*}
  \begin{aligned}
    \mathbb{Q}_3=&\min \sum_{i \in N_1}\sum_{j\in N_2} y_{i,j} \cdot \quad \mathcal{P}_{N_1,N_2}, \\
    s.t.  & \quad 0\leq \sum_{i\in N_1} y_{i,j} \leq 1, \forall j \in \mathbf{F}, \\
    & \quad 0\leq \sum_{j\in N_2} y_{i,j} \leq 1, \forall i \in \mathbf{F}, \\ 
    & \max\sum_{i}\sum_{j}\mathcal{T}_{y_{i,j}}(t),\forall i,j \in \mathbf{F}, t\in T
  \end{aligned}
\end{equation*} 
Note that $i$ and $j$ are from different layers.

\subsection{Average Path Length Approximation}
%从图论的角度出发, 首先分析层内拓扑的平均路径长度
From the perspective of graph models \cite{liu2015analytical}, the average path length of the MLUDSN equates to the average distance of arbitrary two satellites. In this regard, the $1^{st}$ moment and the $2^{rd}$ moment of degree distribution would affect the distance between two randomly chosen satellites \cite{newman2003random,dorogovtsev2008organization}. Specifically, we first exemplify two-layer MLUDSN to exhibit the derivative equations. The adjacency matrix is 
\begin{equation}
  \mathcal{M}^{*}=
  \begin{pmatrix}
    \mathcal{A}_1  & \mathcal{A}_{1,2} \\
    \mathcal{A}_{2,1} & \mathcal{A}_2
  \end{pmatrix}.
\end{equation}

The degree distribution of layer 1 and layer 2 are denoted as $\Gamma_1(q)$ and $\Gamma_2(q)$, respectively. We first investigate the APL of intra-layer topology. The mean number of neighbours of a randomly chosen satellite with degree distribution $\Gamma_1(q)$ is $\mathcal{H}=\left\langle q \right\rangle $, i.e. the $1^{st}$ moment of degree distribution $\mathcal{H}=\sum_q q\cdot \Gamma_1(q)$. Moreover, the mean number of second neighbours of the satellite, i.e. the mean number of other satellites two hops away from the chosen satellite is formulated as 
\begin{equation}
  \begin{aligned}
      \mathcal{H}_2 = & \mathcal{H} \cdot \sum_q q\cdot \Gamma_1(q) \\
      = & \mathcal{H} \cdot \frac{\sum_q q\cdot (q+1)\cdot \Gamma_1(q+1)}{\sum_p p\cdot \Gamma_1(p)} \\
      = & \left\langle q^2 \right\rangle - \left\langle q \right\rangle,
  \end{aligned}
  \label{con-3}
\end{equation}

Here $\sum_q q\cdot \Gamma_1(q)$ is the average number of satellites two hops away from the chosen satellite via a particular one of its neighbours. By multiplying $\mathcal{H}$, we can further obtain $\mathcal{H}_2$. Then we extend to find the mean number of neighbours $x$ hops away from the chosen satellite based on the recursive relations. Thus the average number of neighbours with hops $x$ is
\begin{equation}
  \begin{aligned}
    \mathcal{H}_x &= \frac{\left\langle q^2 \right\rangle - \left\langle q \right\rangle}{\left\langle q \right\rangle} \cdot \mathcal{H}_{x-1} \\
    &=\frac{\mathcal{H}_2}{\mathcal{H}_1} \cdot \mathcal{H}_{x-1},
  \end{aligned}
  \label{rec_1}
\end{equation}
where $\mathcal{H}_1 \equiv  \mathcal{H}$. By iterating this formula, we can obtain the following equation:
\begin{equation}
  \mathcal{H}_x = \left( \frac{\mathcal{H}_2}{\mathcal{H}_1} \right)^{x-1} \cdot \mathcal{H}_1. 
  \label{hops}
\end{equation}
We can observe that Eq. \eqref{hops} would diverge or converge exponentially as $x$ increases, which directly depends on $\mathcal{H}_1<\mathcal{H}_2$ or $\mathcal{H}_1>\mathcal{H}_2$. If $\mathcal{H}_1<\mathcal{H}_2$, it means that the average total number of neighbors of the chosen satellite at all distances is infinite, and there must be a giant component, as shown in Fig. \ref{fig_5}. Note that Eq. \eqref{con-3} can be rearranged as $\left\langle q^2 \right\rangle - \left\langle q \right\rangle =\left\langle q \right\rangle$ when $\mathcal{H}_1=\mathcal{H}_2$, it also indicates that deleting satellites of $q=2$ (i.e., the degree of the satellite is 2) would not change the topological structure of the network, since these satellites fall in the middle of links between other pairs of satellites in layer 1.

The satellites in the giant component, can connect with others by some paths. In fact, Eq. \eqref{hops} also illustrates that the mean number of satellites hops $x$ away from a chosen satellite is based on the giant component. As $x$ equates to the total number of satellites in layer 1 (i.e., $N_1$), $x$ is also the radius of layer 1. Therefore, when $\mathcal{H}_{\overline{\mathcal{D}_1}}\simeq N_1 $, the average satellite-to-satellite distance $\overline{\mathcal{D}_1}$ in layer 1 can be modeled as
\begin{equation}
  \overline{\mathcal{D}_1} = \frac{\ln(N_1/\mathcal{H}_1)}{\ln(\mathcal{H}_2/\mathcal{H}_1)}+1.
  \label{monolayer-1}
\end{equation}
In other words,  $\overline{\mathcal{D}_1}$ can be approximately equals to the APL of layer 1. Similarly, the average satellite-to-satellite distance $\overline{\mathcal{D}_2}$ also depends on the $1^{st}$ and $2^{nd}$ moments of the degree distribution of layer 2.

\begin{figure}
  \centering
  \includegraphics[width=1.2 in]{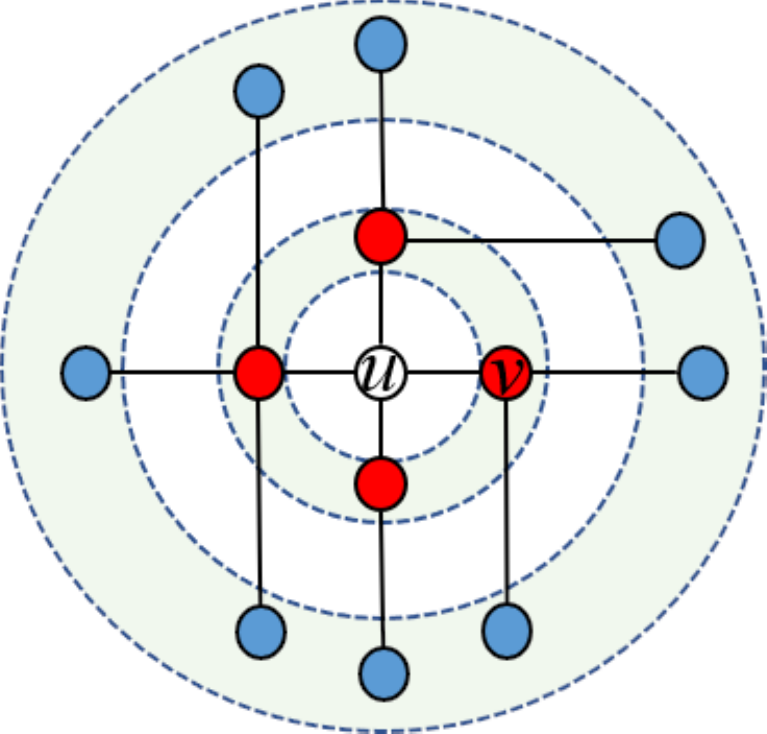}
  \caption{The first (red) and second (blue) connected components of ISL $e_{uv}$, which are tree-like.}
  \label{fig_5}
\end{figure}

%%%%%%%%%%%%%%%%%%%%%%%%%%%%%%%%%%%%
%分析层间的路径长度
Different from the analysis of the APL in the intra-layer topology, the average total distance of the network should consider the inter-degree $p$ and intra-degree $q$ and their joint distribution to derive the total APL. We assume that there are $k$ ILCs without repeated vertices, and $k$ is much smaller than the total number of ISLs in the MLUDSN. Based on the aforementioned analysis, the distribution of component sizes is deeply involved with the total APL. Then the randomly chosen satellite pairs should be classified into four types: $(u,v), (c,d), (u,c)$, and $(v,d)$, where $u,v$ belong to layer 1, and $c,d$ belong to layer 2. We can also obtain the recurrent relationship between $\mathcal{H}_x$ and $\mathcal{H}_{x-1}$ of different satellite pairs. Let $\mathbf{H}$ is denoted as a four-dimensional vectors, thus we first model the case where the component is in layer 1:
\begin{equation}
  \mathbf{H}_x^{1} = \left( \overline{\mathcal{H}}_x^{111}, \overline{\mathcal{H}}_x^{121}, \overline{\mathcal{H}}_x^{112}, \overline{\mathcal{H}}_x^{122}
  \right) 
  \label{H1}
\end{equation}
where the first superscript $\overline{\mathcal{H}}_x^{111}$ indicates whether the component exist in layer 1 or layer 2, the second superscript means whether the start satellite exist in layer 1 or layer 2, and the third superscript means whether the end satellite exist in layer 1 or layer 2. Then we also have the formulation that component is in layer 2:
\begin{equation}
  \mathbf{H}_x^{2} = \left( \overline{\mathcal{H}}_x^{211}, \overline{\mathcal{H}}_x^{221}, \overline{\mathcal{H}}_x^{212}, \overline{\mathcal{H}}_x^{222}
  \right) 
  \label{H2}
\end{equation}

By rearranging these two formulas, we can finally acquire the recursive relation:
\begin{equation}
  \begin{aligned}
      \mathbf{H}_x^1 &= \chi \cdot \mathbf{H}_{x-1}^1 + \mathbf{h}_1,\\
      \mathbf{H}_x^2 &= \chi \cdot \mathbf{H}_{x-1}^2 + \mathbf{h}_2,\\
  \end{aligned}
\end{equation}
where $\mathbf{h}_1=\left(1,1,0,0\right) $ and $\mathbf{h}_2=\left(0,0,1,1\right)$, $\chi$ is a $4\times 4$ matrix, obtained by the generation function of degree distribution. In this regard, we introduce the generation function of the joint degree distribution of layer 1 and layer 2.
Then we can further gain the corresponding distribution functions of the four pairs (i.e., inter-layer and intra-layer), to be the content of $\chi$. The derivation is detailed in Appendix A.
Therefore, we finally obtain the four average sizes of components $x$ hops away from the randomly chosen satellite.
\begin{equation}
  \mathbf{H}_x = A \cdot \chi_1 + B \cdot \chi_2,
\end{equation}
where $A$ and $B$ are the independent coefficients. This equation indicates that the average sizes of the component $x$ hops away from the root satellite, turns out to be linear combinations of $\chi$, and it can be regarded as "mean branch", which depends on the $1^{st}$ and $2^{nd}$ moment of degree distribution of each layer. When $\chi_2 < \chi_1$, it means that the radius of the satellite is smaller than layer 1 and layer 2, and thus the following equation is vaild derived from Eq. \eqref{monolayer-1}:
\begin{equation}
  \overline{\mathcal{D}_{1,2}} = \overline{\mathcal{D}_1} + \frac{\ln(N_2/k)}{\ln \chi_2} + 1,
\end{equation}
where $\overline{\mathcal{D}_{1,2}}$ indicates the average distance between a satellite from layer 1 and a satellite from layer 2. $\chi_2$ the ratio of the $2^{nd}$ moment to the $1^{st}$ moment of degree distribution of layer 2.
Additionally, we can further formulate the total APL of the entire MLUDSN is given as
\begin{equation}
  \mathcal{P}_{N_1,N_2} \simeq \frac{\overline{\mathcal{D}_1}\cdot \mathcal{E}_1 + \overline{\mathcal{D}_2} \cdot \mathcal{E}_2 + \overline{\mathcal{D}_{1,2}} \cdot N_1N_2}{\mathcal{E}_{1,2}},
  \label{apl_1}
\end{equation}
where $\mathcal{E}_{1} = \frac{N_1(N_1-1)}{2}$ is the intra-layer links of layer 1 and  $\mathcal{E}_{2}=\frac{N_2(N_2-1)}{2}$ is the  the intra-layer links of layer 2, respectively. Then $\mathcal{E}_{1,2}=\frac{(N_1+N_2)(N_1+N_2-1)}{2}$ is the total ISLs of layer 1 and layer 2.

%推广到多个网络的加权
Similarly, we further investigate the average distance of the entire network across multiple tandem layers\footnote{Considering the difficulty of establishing cross-layer links, in this paper we only consider the tandem layers and their topologies are isomorphic.}.
Therefore, the total average hop is given as 
\begin{equation}
  \begin{aligned}
    \mathcal{P}_{total} \simeq  \frac{\sum_{i=1}^l \overline{\mathcal{D}}_i \cdot \xi_i + \sum_{i=1}^{l-1}\overline{\mathcal{P}}_{N_i,N_{i+1}}\cdot N_i\cdot N_{i+1}}{\mathcal{E}_{total}}
  \end{aligned}
  \label{apl_2}
\end{equation}
where $\mathcal{E}_{total}=\frac{\sum_{i=1}^l N_i \cdot (\sum_{i=1}^l N_i-1)}{2}$. 

%该节的一个大概总结，先从理论上论证添加层间链是绝对会降低总的平均路径长度的！
\begin{remark}
  The average path length of the entire MLUDSN decreases as the number of ILCs increases. 
  \label{remark}
\end{remark}
\begin{proof}
  Since $\overline{\mathcal{D}}_1$ and $\overline{\mathcal{D}}_2$ only depend on the degree distribution of respective layer, Eq. \eqref{apl_1} infers that the change of total APL lies on $\ln(N_2/k)$, thus we can simplify the equation as
  \begin{equation}
    \mathcal{P}_{N_1,N_2} \propto \ln(\frac{N_2}{k}),
  \end{equation}
  Although $\overline{\mathcal{D}}_1$ and $\overline{\mathcal{D}}_2$ would change as the degree distribution change, the changes are minimal based on the ratio of the $2^{nd}$ moment to the $1^{st}$ moment. Therefore, the APL would decrease as the number of ILCs increases. Then we can also obtain the similar deduction for the total APL of the MLUDSN derived from Eq. \eqref{apl_2}:
  \begin{equation}
    \mathcal{P}_{total}  \propto \{  \ln(\frac{N_2}{k_{1,2}}), \ln(\frac{N_3}{k_{2,3}})...,  \ln(\frac{N_l}{k_{l-1,l}}) \},
  \end{equation}
where $k_{i-1,i}$ indicates the number of ILCs between the $i-1$-th layer and the $i$-th layer. Consequently, the trend of the total APL of the MLUDSN decays logarithmically.
\end{proof}

\section{Two-Phase Deployment Strategy}
%根据问题制定了两阶段算法
In this section, we propose a two-phase inter-layer connection deployment algorithm to decouple the large-scale ILP problem into the two-phase problem, which can significantly reduce the computational complexity of the proposed problem. Specifically, we first construct the feasible set of ILCs based on the physical constraints and the weight constraint, to narrow the feasible region of the problem. Through the narrowed feasible region, we further exploit the genetic method to further select proper ILC sets. Then we design the maximal time weight matching algorithm to search for the most stable topology. By iterating the above operations, we finally can obtain the minimal APL with the minimal number of ILCs.

\subsection{Two-phase Inter-Layer Connection Deployment Algorithm} 
%总算法名称 不用框架描述
Before building the ILCs, the topology of each layer may be determined through four permanent ISLs of each satellite. Many studies illustrate that the satellite topology of a layer is suitable to be constructed as a grid+ topology \cite{9322362,Kuiper2019}. Hence we assume that the intra-layer links have been allocated in this paper. Based on the aforementioned analysis of Section III, the proposed problem $\mathbb{Q}_2$ is still an ILP problem \cite{Floudas2009,guo2016exploiting}, the optimal solution can be acquired by an exhaustive search in the small-scale MLUDSN. However, since satellites deploy densely, the computational complexity has exploded. Thus, we propose a two-phase inter-layer connection deployment (TPILCD) algorithm to further narrow the feasible ILC sets and obtain the optimal matching solutions of multiple candidate ILC sets in parallel, fianlly acquire the optimal ILC deployment strategy after several iterations of convergence.
%%%%%%%%%%%%%%%如何解耦的要说清楚
% 第一步：根据启发式获得局部 多个层间链集合，对这些层间链集合进行交叉和变异操作来进一步，进而并行对多个候选集合进行求解；
%总算法介绍：这里详说可行集的构建；然后算法2说明遗传，算法3说明最优权值匹配
%具体说明一下总算法流程

Specifically, we first construct the physically feasible set, which satisfies that the distance between any two satellites is not greater than the maximum visible distance, and they move in the same direction. We also set the lower threshold and upper threshold of time weight to further filter out ILCs of extremely short or long duration. Then, in the first phase, we first combine the genetic method to propose an optimization of topological layer combinations (OTLC) algorithm, which guarantees the diversity of algorithm results, provides a wider search region for local search, and raises the search breadth of the algorithm. To pursue a stable topology with fewer handover counts of ILCs, we further match ILCs of each set,  with the goal of maximizing the sum of the time weight. Finally, through these operations, we can calculate the total APL with different ILC set, and obtain the minimal $\mathbb{Q}_2$ by several iterations. Notably, the analytical model of network average hop can reduce the complexity of traversing each satellite in a multi-layer disordered network.

\begin{algorithm}[htb]
  \caption{TPILCD Algorithm.}
  \label{alg0}
  \hspace*{0.02in}{\bf Input:} 
  %输入：一个周期的卫星经纬度，海拔高度，
  The longtitude, latitude and altitude of  satellites of each layer over a duration.
  \\
  \hspace*{0.02in}{\bf Output:} 
  The optimal ILC deployment strategy and total APL.

  \begin{algorithmic}[1] %这个1 表示每一行都显示数字
    \STATE Construct the set of buildable ILCs based on Eq. \eqref{con-1}, Eq. \eqref{con-2}, and Eq. \eqref{time}.
    \STATE Randomly Select $\mathcal{J}$ matching sets that each matching set has $k$ ILCs. 
    \STATE \textbf{Repeat}
    \STATE Using OTLC algorithm to operate crossover or mutation for each matching set.
    \STATE Using MTWM algorithm to match the set with maximizing time weight $\mathcal{T}$.
    \STATE \textbf{Until} Reach the maximal iteration counts and the objective converges.
\RETURN  The ILC deployment strategy.%算法的返回值
  \end{algorithmic}
\end{algorithm}

\subsection{Optimization of Topological Layer Combination}
This algorithm concludes with six main steps: initialization; clone ILC set; crossover; mutation; matching; ILC set selection. Specifically, we randomly choose $k/2$ ILCs between layer 1 and layer 2, and these links all locate on the same surface of the Earth \footnote{For example, based on the Eastern and Western hemispheres of the Earth, it is only necessary to consider the ILC matching matrix of satellites deployed on the Eastern hemisphere.}.
The selected ILCs are as one ILC set. In this way, we construct $\mathcal{S}$ ILC sets (Lines 1-4). Then we initialize the number of iterations and initial APL (Line 5). To expand the searching space, $\mathcal{CS}$ ILC sets are cloned (Line 7).

The crossover operation can be triggered with probability $P_c$ when traversing each ILC set. Then we randomly select $p_j^b$ as the intersecting ILC set with $p_j^a$. We randomly select the cross position in the ILC set, and cross ILCs before the cross position. Other ILCs remain unchanged. Different from the crossover operation that occurs between multiple ILC sets, the mutation operation carries out within an ILC set. This operation may be triggered with probability $1-P_c$. In this operation, we reshuffle the satellites of layer 1 and layer 2 in a randomly selected ILC set. Moreover, this ILC set constructs new ILCs based on the newly sorted satellites and put them into the new ILC set. Note that not all new ILCs can be established. Due to the randomness of sorting, the construction of some ILCs may not satisfy the time weight constraint. Thus, the newly generated ILCs should be determined whether they can be constructed practically. If not, the crossover or mutation operation should be restarted until all ILCs satisfy the constraints (Lines 9).

% \begin{figure}
%   \centering
%   \includegraphics[width=2.8 in]{fig6.eps}
%   \caption{Crossover operation for population $p^1$ and $p^2$.}
%   \label{fig_6}
% \end{figure}

%说明链路资源均衡
Considering the limited antenna resources in each satellite, we count the degree of the satellites $deg_{u,v}$ that constitute the ILC and the average degree of their neighbor satellites $deg_{u,v}^{neig}$ by traversing each ILC set. Based on the degree of prior knowledge, we can construct a new ILC, and randomly delete an old ILC that is not composed of hub satellites. Thus, the total APL can be calculated by adjacency matrixes and ILCs. Meanwhile, when the obtained APL decreases, the new generated ILC is stored in the set, otherwise the old ILC is retained (Line 10).  

To pursue the topology stability over a long duration, we further consider that how to reduce the number of handovers, aiming at maintaining more stable connections between layers, and reducing the cost of handovers and the complexity of routing scheme. Based on Eq. \eqref{time_con}, we aim to  update the connections of satellites and to maximize the sum of weight of the candidate set (Line 11). Meanwhile, we extend $\mathcal{S}$ as $\widehat{\mathcal{S}}$ based on the symmetric topology and acquire the candidate set.
Derived from the matching matrix $\mathcal{A}_{1,2}$ constituted by the candidate set, we can calculate the APL from different sets, and the set with the minimal APL is regarded as the next iteration set (Lines 12-13). Then we totally select $\mathcal{S}$ sets to constitute $\mathcal{S}^{*}$. The iteration terminates until  the APL convergence (Lines 15-20).

\subsection{Maximizing Time Weight Matching}
Since the feasible region is further narrowed by operations such as crossover and mutation, we consider a candidate set consisting of $k$ ILCs, and the objective is to find the optimal matching matrix with the maximal sum of the time weight, which is modeled as follows:
\begin{equation}
\begin{aligned}
  \mathbb{Q}_4: y^{*} &=\text{argmax} \sum_{i=1}^k \sum_{j=1}^k \mathcal{T}_{y_{i,j}}(t), \quad\forall t\in T,\\
s.t. & \quad \sum_{i=1}^k y_{i,j} \leq 1, \forall j=1,2,...,k, \\
    &  \quad \sum_{j=1}^k y_{i,j} \leq 1, \forall i=1,2,...,k, 
\end{aligned}
\end{equation}
where $k$ is the number of ILCs. To solve this subproblem, we build a bipartite graph, where the vertices in the left are from layer 1, and the vertices in the right are from layer 2. We first initialize each satellite both in the left and right by the topmark, denoted as $LA(i)$ and $LB(j)$, respectively (Line 1). Then the subproblem is transformed into finding perfect matching problem, which should satisfy the following condition for arbitrary a link $e_{ij}$:
\begin{equation}
  LA(i)+LB(j)\geq\mathcal{T}_{y_{i,j}}(t), \forall i,j \in [1,..,k], t\in T.
\end{equation} 
By finding the augmentation path and modifying the topmark (Lines 2-7), we can finally provide a perfect matching matrix to compute the APL of total network.

% \begin{figure}
%   \centering
%   \includegraphics[width=1.5 in]{fig8.eps}
%   \caption{A bipartite graph with the time weight of one candidate set.}
%   \label{matching}
% \end{figure}

%计算复杂度的分析
\subsection{Complexity Analysis}
Considering the OTLC algorithm is the core of the TPILCD algorithm, we first analyze the complexity of its main steps. We focus on the crossover, mutation, link resource balancing, and ILC selection, since they deeply influence the number of iterations. Specifically, the complexity of crossover operation, mutation operation, link resource balancing, and ILC selection are $\mathcal{O}(\mathcal{S}\cdot \mathcal{CS})$, $\mathcal{O}(k \cdot \mathcal{S}\cdot \mathcal{CS} )$, $\mathcal{O}(k\cdot \mathcal{S}\cdot \mathcal{CS} \cdot \Upsilon )$, and $\mathcal{O}(\mathcal{S}\cdot \mathcal{CS}  \cdot \Upsilon)$, respectively. Here $\mathcal{S}$ is the number of candidate ILC sets, $\mathcal{CS}$ is the number of clone sets, $k$ is the number of ILCs and $\Upsilon$ is the time consumption of computing the total APL of the entire MLUDSN. 

\begin{algorithm}
  \caption{OTLC Algorithm.}
  \label{alg1}
  \hspace*{0.02in}{\bf Input:} 
  Adjacency matrix of two-layer satellite network $\mathcal{A}_1$, $\mathcal{A}_2$; the number of ILCs $k$; the number of ILC sets $\mathcal{S}$; the number of clone sets $\mathcal{CS}$; crossover Probability $P_c$; mutation Probability $1-P_c$; Iteration Number $it_{max}$ \\
  \hspace*{0.02in}{\bf Output:} 
  The optimal ILC deployment strategy and the minimal APL
\begin{algorithmic}[1] %这个1 表示每一行都显示数字
  \STATE Based on the symmetry of the topology of each layer, we first set the half ILCs ($k/2$) that locate on the same surface of the Earth.
  \FOR{each ILC set $p_j\in \mathcal{S}$}
   \STATE Randomly generate ILC $e_{uv}$ with non-repeating satellites, $u\in \mathcal{A}_1, v \in \mathcal{A}_2$, then add $e_{uv}$ into $p_j$. 
   \ENDFOR
   \STATE Initialize $it=0$ and $\overline{\mathcal{P}}_{N_1,N_2}=\inf$
   \WHILE{$it<it_{max}$}
   \STATE Copy $p_j \to p_j^a$, and add $p_j^a$ into $\mathcal{S}^{*}$
   \FOR{ each ILC set $p_j^a\in \mathcal{S}^{*}$}
   \STATE Operate crossover or mutation by $P_c$.  
   \STATE Find the neighbor $p_j^{a'}$ based on the degree prior knowledge.
   \STATE \emph{Using MTWM algorithm} to obtain candidate ILC sets.
   \STATE Based on the logical index of satellties, the total ILC matching matrix $\mathcal{A}_{1,2}$ can be obtained.
   \STATE Calculate $\overline{\mathcal{P}}_{N_1,N_2}$ by Eq. \eqref{apl_1} and add optimal $p_j^{a}$ into next iteration ILC set $\mathcal{S}^{*}$.
   \ENDFOR
   \IF{$\overline{\mathcal{P}}_{N_1,N_2}$ is equal to last iteration $\overline{\mathcal{P}}_{N_1,N_2}$}
   \STATE $it=it+1$.
   \ELSE
   \STATE $\overline{\mathcal{P}}_{N_1,N_2}$ is the optimal solution and generate the candidate ILC set.
   \ENDIF
   \ENDWHILE
\end{algorithmic}
\end{algorithm}

If we adopt the all-to-all transit mode to obtain the shortest paths of all pairs of satellites, then $\Upsilon= \mathbf{N}^3$. The reason is that in addition to traversing all satellite pairs, each satellite needs to be traversed once to determine whether the satellite is a relay vertice. Therefore, the total complexity of the OTLC algorithm under the all-to-all transit mode is $\mathcal{O}(k \cdot \mathcal{S}\cdot \mathcal{CS} \cdot \mathbf{N}^3)$. Instead of the traversing method, the complexity of the proposed APL formulation is $\mathcal{O}(\mathbf{N})$ since the degree distribution of each satellite in the total network needs to be previous obtained, i.e., $\Upsilon= \mathbf{N}$. Thus, the total complexity of the OTLC algorithm is $\mathcal{O}(k \cdot \mathcal{S}\cdot \mathcal{CS} \cdot \mathbf{N} \cdot \Phi )$, where $\Phi$ is the time consumption of matching the candidate set during the OTLC algorithm.  
%分析最大权值匹配算法的复杂度
Thus, we further explore the complexity of the MTWM algorithm, which depends on the number of ILCs. Its computational complexity is $\mathcal{O}(k^3)$, i.e., $\Phi=k^3$. In short, the total complexity of TPILCD algorithm is $\mathcal{O}(k^4 \cdot \mathcal{S}\cdot \mathcal{CS} \cdot \mathbf{N}^3)$. Compared with the exponential complexity of the original problem, the complexity of the algorithm is reduced to polynomial.

%% 第五部分：仿真结果
\section{Simulation Analysis}

In this section, our simulations are conducted on several scenarios, which are based on a real satellite network by the satellite orbit parameters supplied by the satellite tool kit (STK). To compare the performance of the proposed algorithm and the ILP model which is feasible in the small-scale network, we first perform simulations in two-layer small-scale networks, which consist of the Globalstar constellation and the Celestri constellation. Then we focus on the interconnections of the Kuiper constellation, to investigate the impacts of network capability, and ILCs on different network performance indicators. Especially, we also simulate the SpaceX constellation to explore the combination of the layer with polar orbits and the layer with inclined orbits. Finally, we also introduce the global traffic to verify the performance of the proposed algorithm. The time horizon is from 2022-8-1 10:00:00 to 2022-8-1 12:00:00 with a duration of 60 seconds of a time slot. The specific simulation parameters are present in Table \ref{tab1}. All the simulations are performed on MATLAB and Pycharm in a PC with i5-11300H CPU 3.1GHz, 16GB RAM.

\begin{algorithm}
  \caption{MTWM Algorithm.}
  \label{alg2}
  \hspace*{0.02in}{\bf Input:} 
  %输入：一个周期的卫星经纬度，海拔高度，
 The candidate ILC set, the corresponding time weight matrix.
  \\
  \hspace*{0.02in}{\bf Output:} 
  The ILC set with the maximal sum of time weight.

  \begin{algorithmic}[1] %这个1 表示每一行都显示数字
    \STATE \textbf{Initializion} Set the selected satellite topmark in layer 1 to the maximum link weight of all ILCs, and set the selected satellite topmark in Layer 2 to 0.
    \STATE \textbf{Repeat}
    \STATE Exploit the Hungarian algorithm to find the perfect matching set.
    \IF{A perfect matching set is not found}
    \STATE Modify the vaule of the feasible topmark.
    \ENDIF
    \STATE \textbf{Unitl} Obtain the perfect matching set of ILCs. 
\RETURN  The ILC deployment strategy.%算法的返回值
  \end{algorithmic}
\end{algorithm}

% \begin{table}[]
%   \centering
%   \caption{Simulation Parameters \cite{wang2007topological,9351765}}  
%   %\resizebox{\columnwidth}{!}{%
%   \begin{tabular}{l|cccc}
%   \hline
%              & \emph{\textbf{Planes}} & \emph{\textbf{Sat/plane}} & \emph{\textbf{Altitude}} & \emph{\textbf{Inclination}} \\ \hline
%              \emph{\textbf{Globalstar}} & 8      & 6         & 1414     & 52          \\
%              \emph{\textbf{Celestri}}   & 7      & 9         & 1400     & 48          \\
%              \emph{\textbf{Kuiper-B}}   & 36     & 36        & 610      & 42          \\
%              \emph{\textbf{Kuiper-C }}  & 28     & 28        & 590      & 33          \\ 
%              \emph{\textbf{SpaceX-1}} & 72 &22 & 550   & 53 \\ 
%              \emph{\textbf{SpaceX-2}} & 6 &58  &560   & 97.6 \\ \hline
%   \end{tabular}%
%   %}
%   \label{tab1}
% \end{table}

\begin{table}[]
  \centering
  \caption{Simulation Parameters \cite{wang2007topological,9351765}} 
  \resizebox{\columnwidth}{!}{%
  \begin{tabular}{|l|l|l|l|l|}
  \hline
             & Planes & Sats/plane & Altitude(km) & Inclination($^\circ$) \\ \hline
  Globalstar & 8      & 6          & 1414         & 52             \\ \hline
  Celestri   & 7      & 9          & 1400         & 48             \\ \hline
  Kuiper-B/C     & 28/36  & 28/36      & 590/610      & 33/42          \\ \hline
  SpaceX-1/2/3    & 72/6/36   & 22/58/20      & 550/560/570      & 53/97.6/70        \\ \hline
  \end{tabular}%
  } 
  \label{tab1}
\end{table}

% \begin{figure}
%   \centering
%   \includegraphics[width=2.7 in]{s1_1.eps}
%   \caption{Average path length comparison among the optimal solution, TPILCD algorithm and greedy algorithm (GA).}
%   \label{s1}
  
% \end{figure}
% \begin{figure}
%   \centering
%   \includegraphics[width=2.3 in]{s1_2.eps}
%   \caption{Total cost comparison among three algorithms.}
%   \label{s3}
% \end{figure}

\begin{figure}
  \centering
  \subfloat[]{\label{s1_1}\includegraphics[width=2.7 in]{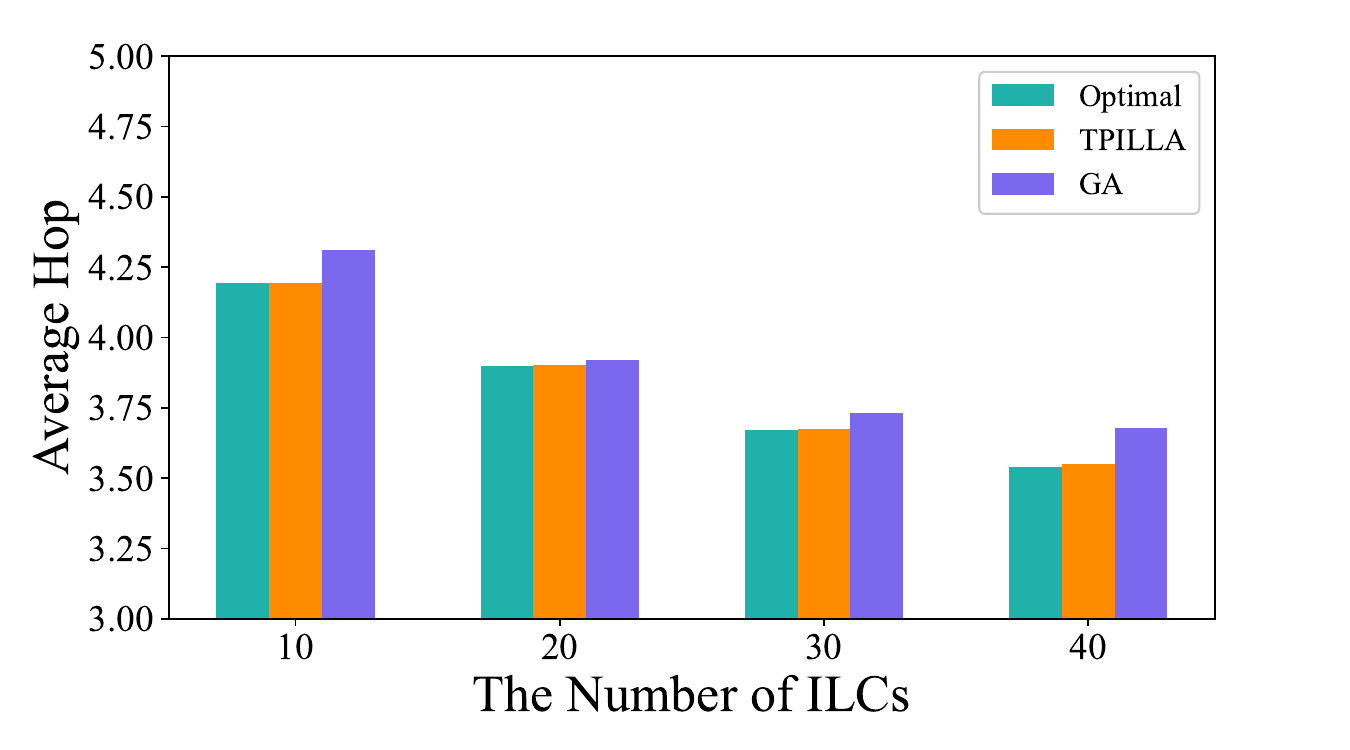}} 
  \\
  \vspace{-1em}
  \subfloat[]{\label{s1_2}\includegraphics[width=2.7 in]{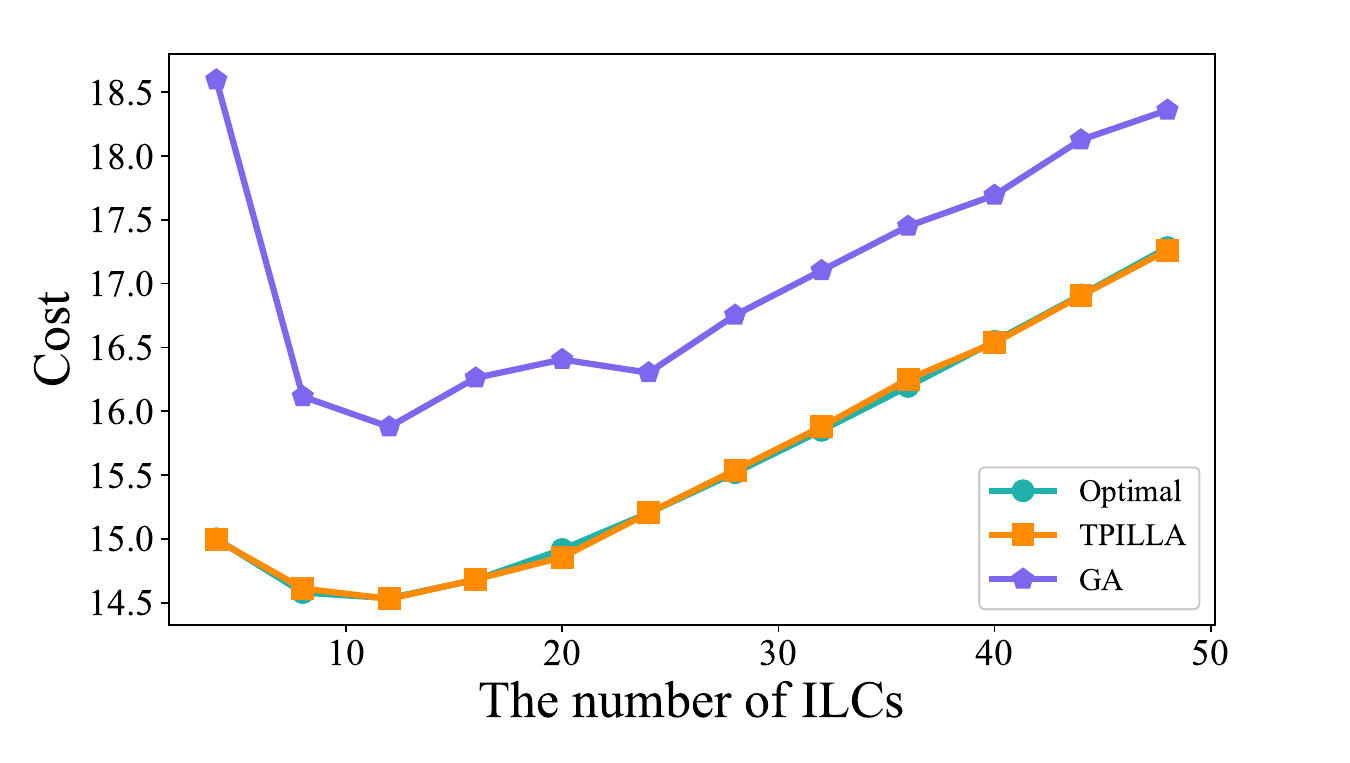}}
  \caption{Performance comparison among the optimal solution, TPILCD algorithm and greedy algorithm (GA). (a) Average hop (b) Cost.}
\end{figure}

\subsection{Performance Evaluation}

%首先说明小规模下与最优性能的对比，包括：验证平均跳数是否递减，以及综合开销最小的层间链数量为多少，
To verify the performance of the proposed algorithm, we first build ILCs between the Globalstar constellation and the Celestri constellation. The reason for choosing these two constellations is that the GlobalStar constellation has an even number of orbits and the number of satellites per orbit, while the Celestri constellation has an odd number of them, and ILCs between them are of generality. As shown in Fig. \ref{s1_1}, with the increase of ILCs, the three APLs obtained by the optimal solution, TPILCD algorithm, and greedy algorithm respectively both decay logarithmically. Here the optimal solution is acquired by traversing all the pairs of satellites by the CVX tool. We observe that the APL obtained by the TPILCD algorithm is almost equal to the optimal solution, and the error appears as the number of ILCs up to 40, which results from the rare inaccuracies of ILCs of the TPILCD algorithm and the approximate model of the total APL. Compared with the TPILCD algorithm, the greedy algorithm builds ILCs with the shortest distance and neglects the impacts of ILCs on the total network performance. In Fig. \ref{s1_2}, we mainly exhibit the distribution of satellites that build ILCs in the GlobalStar constellation, since the maximal number of ILCs is constrained by the GlobalStar constellation. 
The ILCs of the proposed algorithm is mainly concentrated in the middle zone of the grid topology, while the greedy algorithm is more concentrated in the margin. As a result, the cross-layer routing needs to traverse to the margin satellite first, which imposes the redundant path, resulting in a large APL.

\begin{figure}
  \centering
  \includegraphics[width=2.3 in]{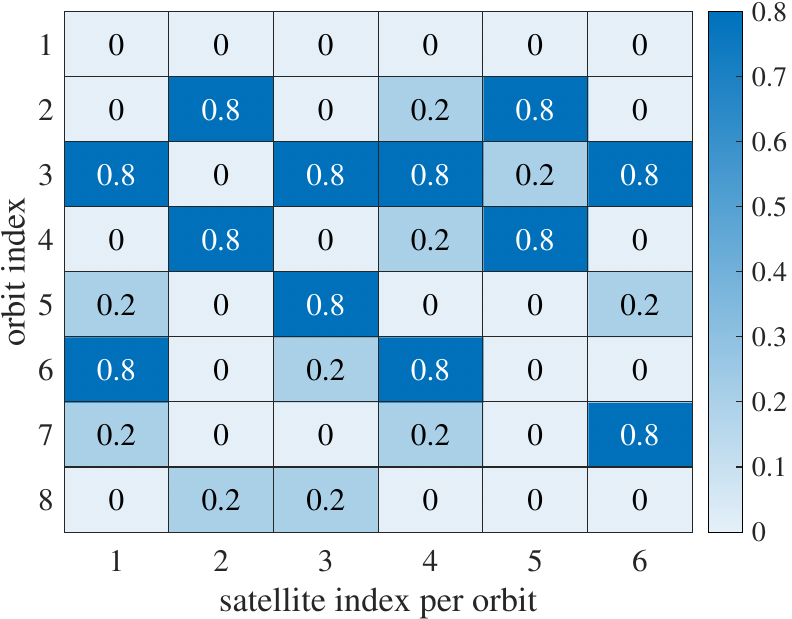}
  \caption{The distribution of satellites that build ILCs in the GlobalStar constellation (the weight=0.8 indicates the selected satellite from the TPILCD algorithm, and the weight=0.2 indicates the chosen satellite from the greedy algorithm, the weight=0 indicates there is no satellite to build the ILC). The difference in weights is only to distinguish the satellites that build the ILCs under different algorithms.}
  \label{s2}
\end{figure}

%进一步说明开销（先说明对开销的处理）
Furthermore, we compare the total cost among the aforementioned three algorithms. Since the objective is the product of the number of ILCs and the total APL. Considering that the former is likely to be far greater than the latter, we normalize the link counts into a range similar to the APL, as presented in Fig. \ref{s2}. It can be seen that blindly adding ILCs may impose the total cost sharply, and waste the abundant link resources. Meanwhile, the error between the cost of the proposed algorithm and the optimal solution is 1\%. Although the cost from the greedy algorithm is larger than the proposed algorithm, the trend of the cost curve also illustrates that the performance gain of the ILCs is limited, and the optimal number of ILCs is less than half the size of the GlobalStar (i.e. ILCs=24). Thus we can derive the following theorem based on the objective:

%说明是小于较小层的规模的一半
\begin{theorem}
  The number of ILCs should be less than half the size of a layer that has a smaller scale, i.e., $k<\frac{N_2}{2}$ if $N_1>N_2>\ln N_1+3$.
  \label{thm-1}
\end{theorem}

\begin{proof}
Assuming that $N_1>N_2$, i.e. $\overline{\mathcal{D}_1} > \overline{\mathcal{D}_2}$, and the maximal number of ILCs is $N_2$. Due to the normalization of ILCs, we assume the corresponding cost of ILCs is $\overline{\mathcal{D}_1}$, as shown in Fig. \ref{prf}. Then we can obtain two functions of the objective:
\begin{equation}
  \begin{aligned}
      f_1&=\overline{\mathcal{D}_1} + \frac{ln(N_2/k)}{ln\chi_2}+1, \\
      f_2&=\frac{\overline{\mathcal{D}_1}}{N_2}\cdot k.
  \end{aligned}
\end{equation}
Thus the objective is $f=f_1\cdot f_2$. Then we differentiate it to find a first order condition:
\begin{equation}
  f^{'} = \frac{\overline{\mathcal{D}_1}^2}{N_2}+\frac{\overline{\mathcal{D}_1}}{N_2}\cdot \frac{\ln(N_2/k)}{\ln \chi_2} + \frac{\overline{\mathcal{D}_1}}{N_2}-\frac{\overline{\mathcal{D}_1}}{\ln \chi_2}=0.
\end{equation}
Based on the inequation of $\ln x \leq x-1$, we have
\begin{equation}
  k \leq \frac{N_2\cdot \ln\chi_2}{N_2-\overline{\mathcal{D}_1}\cdot \ln\chi_2}.
\end{equation}
Considering the grid topological feature of each layer, the average degree is at [3,5], and $\chi_1<2$, $\chi_2<2$ absolutely. By exploiting the backward deduction, we first assume that $  k \leq \frac{N_2\cdot \ln\chi_2}{N_2-\overline{\mathcal{D}_1}\cdot \ln\chi_2}< \frac{N_2}{2}$ is permanent, then we set $\chi_1=\chi_2=2$, and the inequation is transformed as follows:
\begin{equation}
  N_2-\ln N_1 > 4-\ln H_1,
\end{equation}
where $H_1$ is the average degree of layer 1. Due to $\ln H_1 \in [1,1.6]$, and the corresponding condition can be obtained as $N_2>\ln N_1+3$. 

\end{proof}

\begin{figure}
  \centering
  \includegraphics[width=2.0 in]{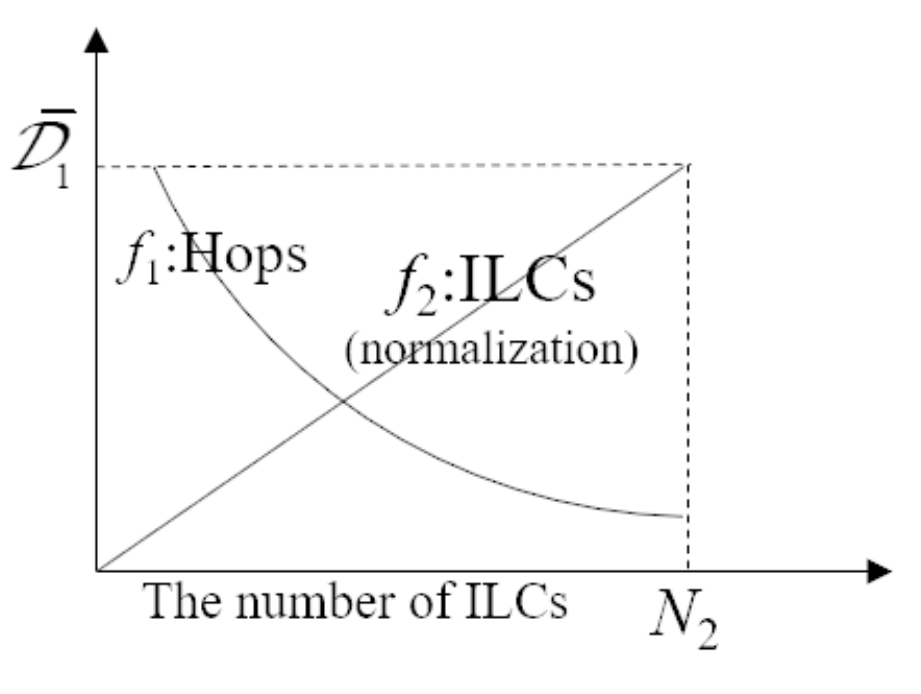}
  \caption{An illustration of two functions: $f_1$: the average hop, $f_2$: the normalized cost of ILCs.}
  \label{prf}
\end{figure}

%Firstly, we formulate the link rate by Shannon's theorems \cite{9327501}: 
% \begin{equation}
%   \mathcal{R}_{u,v}= B\cdot \log (1+SNR_{u,v}),
%   \label{rate1}
% \end{equation} 
% where $B$ is the channel bandwidth in Hertz and $SNR_{u,v}$ is the signal-to-noise ratio for an ongoing transmission from the satellite $u$ to $v$ and vice-versa. Then the $SNR$ is presented as
% \begin{equation}
%   SNR_{u,v} = \frac{P_t\cdot G_{0}^2}{k_B\cdot\tau\cdot B \cdot F_{uv}},
%   \label{rate2}
% \end{equation}
% where $P_t$ is transmission power, $G_0$ is the maximum gain of transmitter and receiver antennas, $k_B=1.381\cdot 10^{-23}$ is the Boltzmann constant, $\tau=354.18$ is the thermal noise in Kelvin, and $F_{uv}$ is the free-space path loss between $u$ and $v$, that is,
% \begin{equation}
%   F_{uv} = \bigg(\frac{4\cdot \pi \cdot d_{uv}\cdot f}{c}\bigg)^2.
%   \label{rate3}
% \end{equation}
% Here $f$ is the carrier frequency and $c=3\cdot 10^8$ m/s is the speed of light. Through Eq. \eqref{rate1}-Eq. \eqref{rate3}, we set $B=20$MHz, $P_t=3.74$w, and the maximal link rate is calculated about 1 Gbps. 
%重写！！！ 讲清楚关系

%计算链路容量（舍去）
Based on the origial problem $\mathbb{Q}_1$, we can further obtain throughput by the known topologies, matching matrix, and link rate.  We set the channel bandwidth as 20MHz and transmission power as 3.74w, thus the maximal link rate is calculated about 1 Gbps \cite{9327501}. By adopting an all-to-all traffic model, we sum all the link rates, and assess the maximum throughput (Throughput: GlobalStar-47Gbps; Celestri-55.8Gbps; Two-Layer-117.984Gbps). Compared with two independent layers, two-layer throughput can be augmented by 21.87\% with the added 12 ILCs. Even the greedy algorithm also enhances 11.58\% throughput. These two results also verify that the MLUDSN raises throughput more than multiple monolayer satellite networks by tolerating little extra overhead.

\subsection{Multi-Layer Collaboration Gain of Mega Constellations}

%先分析相同的趋势，也就是跳数的递减，和最佳链路数量的公式论证，以及容量增益
In this subsection, we further investigate the performance of the proposed algorithm and the impacts of network connection capability in mega constellations on network performance. Considering that the original problem with the exhaustive search is intractable as the huge network scale, we mainly compare the proposed algorithm with the greedy algorithm and random algorithm. Based on the detailed Kuiper and SpaceX constellations in Table \ref{tab1}, we focus on the combination of two layers with inclined orbits (Kuiper), and the combination of layers with hybrid orbits, i.e., polar orbits and inclined orbits (SpaceX). 
Fig. \ref{k1} firstly verifies that the total APL decays logarithmically with the increase of the number of ILCs. The difference between two constellations is that: $P$ is about 3.3 times larger than $S$ in SpaceX-1, and $P$ is about 9.7 times smaller than $S$ in SpaceX-2, which renders the dramatical decrease of the E2E hops by inter-layer paths. Even though the scale of the SpaceX is similar to the Kuiper ($\mathbf{N}_{spaceX}=1932,\mathbf{N}_{kuiper}=2080$), the gain of two-layer collaboration under the SpaceX topological structure is more outstanding.  Fig. \ref{k2} also presents a huge cost decrease in costs generated by SpaceX's multi-layer collaboration than the one of the Kuiper, and we infer that the optimal number of ILCs is 300 and 160 respectively. 
%Table \ref{table2} further exhibits the gain of the cooperative work with two layers, where the Kuiper constellation can raise 23.78\% capacity, and the SpaceX constellation can raise 55.4\% capacity. 
%对比观察到若P和S相差较大时，通过层间连接可更大程度提升网络指标

%改成图展示容量变化（1）随层间链数量的增多，吞吐量的增益有限

%其次，分析Kuiper星座和SpaceX星座构建层间链的差异性，也就是构型本身的影响

% The reason is that the APL is actually a ratio of the total number of hops of all the satellite pairs and the number of satellite pairs. The combination of these two layers imposes a significant improvement of the numerator than the denominator. As the inter-layer links rise, the numerator and denominator increase in the opposite trend, so that the total APL gradually drops between the hop counts of the two monolayers. Besides, since the greedy algorithm referentially selects the satellite pair with the shortest distance, these inter-layer links are concentrated in the margin of the grid+ topology, which increases the end-to-end hop count of the cross-layer satellite pairs and further imposes the total APL. Different from the greedy algorithm, the random algorithm randomly selects several satellites in each orbit in the feasible sets to build inter-layer links, which can deploy the chosen satellites distributionally, and obtain the approximate performance from that of the proposed algorithm. However, the error between them would be magnified after counting the total cost.

%kuiper-fig1: 跳数展示，说明在大规模下的适用性, 且验证了跳数随层间链数量增加后的对数衰减
\begin{figure}
  \centering
  \includegraphics[width=2.7 in]{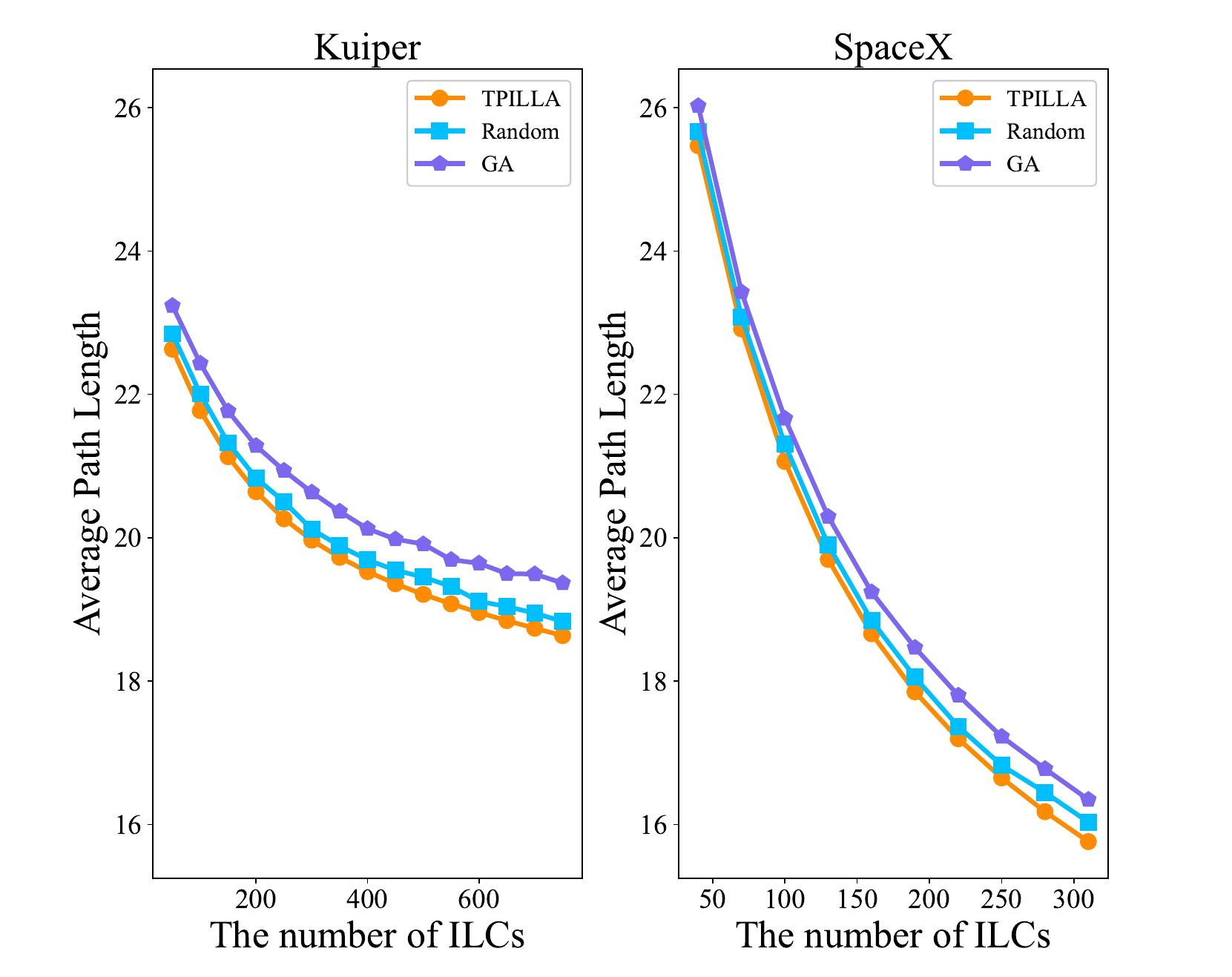}
  \caption{Total average hop comparison among three algorithms in Kuiper and SpaceX constellations.}
  \label{k1}
\end{figure}

Furthermore, since the greedy algorithm referentially selects the satellite pair with the shortest distance, these ILCs are concentrated in the margin of the grid+ topology, which increases the E2E hop count of the cross-layer satellite pairs and further imposes the total APL. Different from the greedy algorithm, the random algorithm randomly selects several satellites in each orbit in the feasible sets to build ILCs, which can deploy the chosen satellites distributionally, and obtain the approximate performance from that of the proposed algorithm. However, the error between them would be magnified after counting the total cost.
Derived from the APL obtained by each algorithm, the cost of the greedy algorithm is the largest, though it consumes minimal power based on the shortest distance\footnote{The power of the inter-satellite link is positively related to the square of the distance}. Physically, the proposed algorithm is 56.21\% (Kuiper) and 81.13\% (SpaceX) more cost-efficient than the greedy algorithm, as well as 18.1\% (Kuiper) and 41.9\% (SpaceX) than the random algorithm, after the normalization of these costs. 
% 添加网络吞吐量随层间链的变化
As shown in Fig. \ref{t1}, we also explore the gain under multi-layer collaboration for throughput. Although the inter-layer communication resources are increased, total throughput still presents a trend of first growth and then stability, which illustrates that the cooperation gains are not a simple linear addition. Besides, we further investigate the impacts of the ILC deployment region on network indicators. Fig. \ref{t2} shows that under the same number of ILCs, with the expansion of the deployment region,distributed deployment can enhance the throughput by 8.1\% and 4.29\% with SpaceX and Kuiper, respectively, as well as reduce the hop by 3.39\% and 2.39\%, compared with centralized deployment. This is because there are more options for E2E paths instead of occupying the same links repeatedly, thus improving the link utilization and total throughput.

% As presented in Fig. \ref{k2}, we follow the standardization of the number of inter-layer links in the previous subsection, and drop them into a range [18,24]. Then we product it and the APL to acquire the total cost. Basically, all the algorithms exhibit the same tendency of the cost, and infer that the optimal number of inter-layer links is 300. 

%kuiper- fig2：展示cost的变化以及对比算法，说明
\begin{figure}
  \centering
  \includegraphics[width=2.7 in]{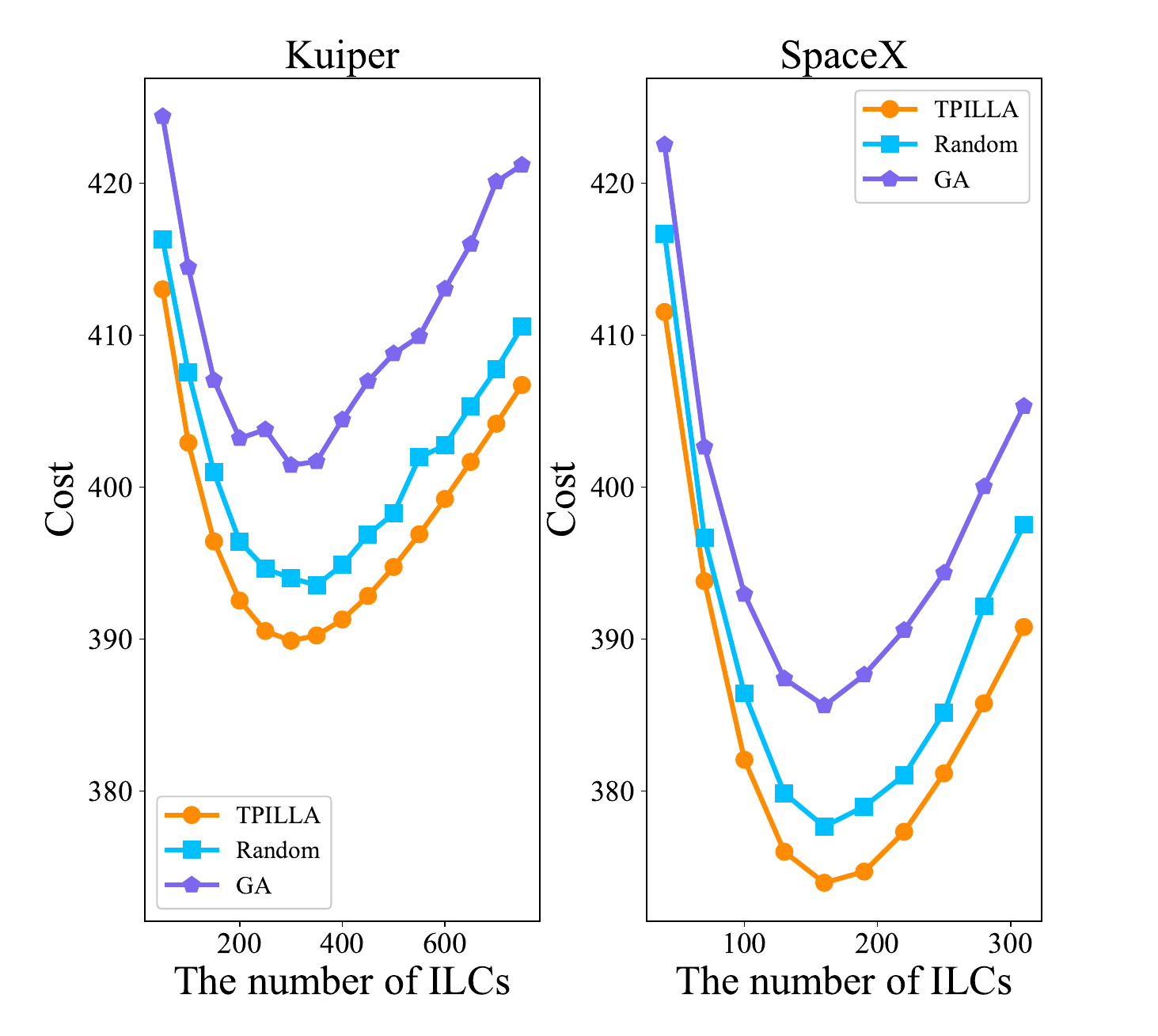}
  \caption{Total cost comparison among three algorithms in Kuiper and SpaceX constellations.}
  \label{k2}
  \vspace{-1em}
\end{figure}

\begin{figure}
  \centering
  \includegraphics[width=2.8 in]{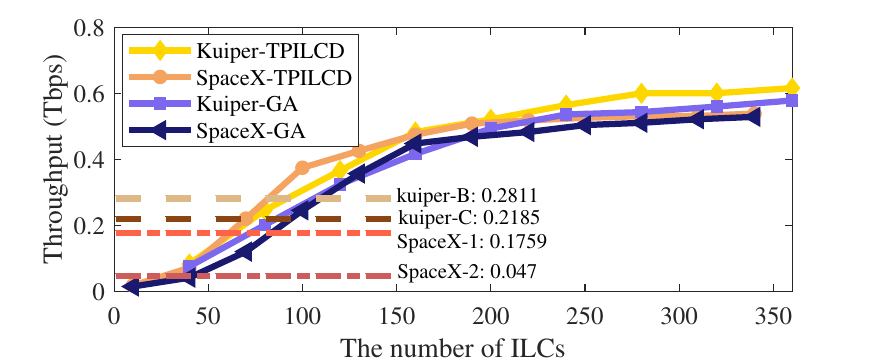}
  \caption{Throughput comparison in Kuiper and SpaceX constellations.}
  \label{t1}
  \vspace{-1em}
\end{figure}

\begin{figure}
  \centering
  \includegraphics[width=2.8 in]{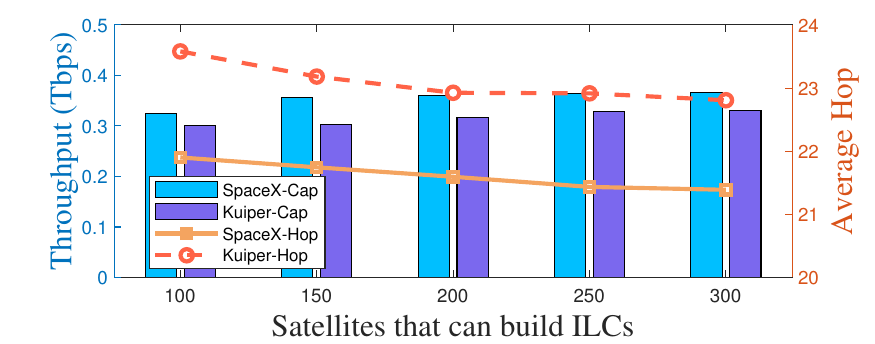}
  \caption{Throughput and hop change with different satellites that can build ILCs (SpaceX-real-ILC=100, Kuiper-real-ILCs=100).}
  \label{t2}
  \vspace{-1em}
\end{figure}

%统计切换次数
Additionally, we also count the number of handovers over 60 time slots to illustrate the stability of ILCs in Fig. \ref{k3}. Note that the X-axis is the maximum number of ILCs that can be connected per orbit, which can be regarded as the network connective capability. Generally, the number of handovers rises first and tends to stabilize, which indicates that the number of handovers can be controlled by adjusting the number of satellites that can establish the ILC on each orbit.
Moreover, since the greedy algorithm exploits the shortest distance to establish ILCs, ILCs are pretty unstable and the handover frequency is significantly high. Thus, the number of handovers of the greedy algorithm is nearly $2\times$ larger than other algorithms. Although the random algorithm deploys the distributionally ILCs in the total topology, it still ignores the influence of the link duration on the handover process, the total number of handovers is still larger than that of the proposed algorithm. We also introduce the \emph{MaxTimeWeightSum} algorithm to display the lower bound of the handovers. This algorithm focus on maximizing the sum of the time weight of all ILCs instead of taking into account other costs.

% To further explore the impacts of network connective capability on the network topology. We introduce \emph{natural connectivity} to evaluate network invulnerability. \emph{natural connectivity} depicts the redundancy of the alternative path through the weighted sum of the number of closed loops over different paths \cite{com2014}, which is given as
% \begin{equation}
%   \overline{\mathcal{NC}} = \ln(\frac{\sum_i^{\mathbf{N}} e^{\lambda_i}}{\mathbf{N}}),
% \end{equation}
% where $\lambda$ is the characteristic root of the adjacency matrix of the total topology. Note that the greater natural connectivity, the stronger network invulnerability.
% Fig. \ref{k4} highlights two points. One is that increasing link resources would enhance network invulnerability, the other is that the limited network connection capability would deploy ILCs distributionally in the feasible region, which raises the network redundant paths and finally enhance network invulnerability.

% %kuiper-fig4:网络能力测试，将每轨可建链的数量限制后，对跳数和自然连通度的影响
% \begin{figure}
%   \centering
%   \includegraphics[width=2.3 in]{k_4.eps}
%   \caption{Total handover counts comparison among three algorithms in the Kuiper constellation.}
%   \label{k4}
% \end{figure}

We also present the real distribution of ILCs in the Kuiper-C, as shown in Fig. \ref{k5}. The selected satellites are centrally symmetric at latitude 0, and most interlayer chains are deployed at latitude [-33,-20], [20,33], indicating that the selected satellite pairs move up or down at the same time in the manner of crossing orbit. Meanwhile, Fig. \ref{k7} also verifies that the maximum number of satellites connected per orbit is less than $S/2$ and ILCs are almost uniformly allocated on each orbit. Besides, we also exhibit the geographical distribution of the satellites that established ILCs on the SpaceX-1.
Most ILCs are preferentially allocated with the start satellite and the end satellite in orbit at high and low latitudes.
%最后，对不同P/S比情况下不同层的增益趋势分析，强调
Finally, we explore the intertwined impacts of constellation's structures and the total cost on throughput, as shown in Fig. \ref{imp1}. We mainly take the Kuiper-C as a fixed layer and observe the gain of ILC deployment by coupling other constellations of different structures, where parameters $P$ and $S$ are (28,28), (36,36), (36,20), (72,22), respectively. On the one hand, throughput gain is dramatic with a large $P/S$, since the average hop reduces sharply. On the other hand, adding a few ILCs decreases throughput because the E2E hops rise significantly with the additional satellite pairs.

%kuiper-fig3:切换次数展示，说明时间权重引入的重要性

\begin{figure}
  \centering
  \includegraphics[width=2.8 in]{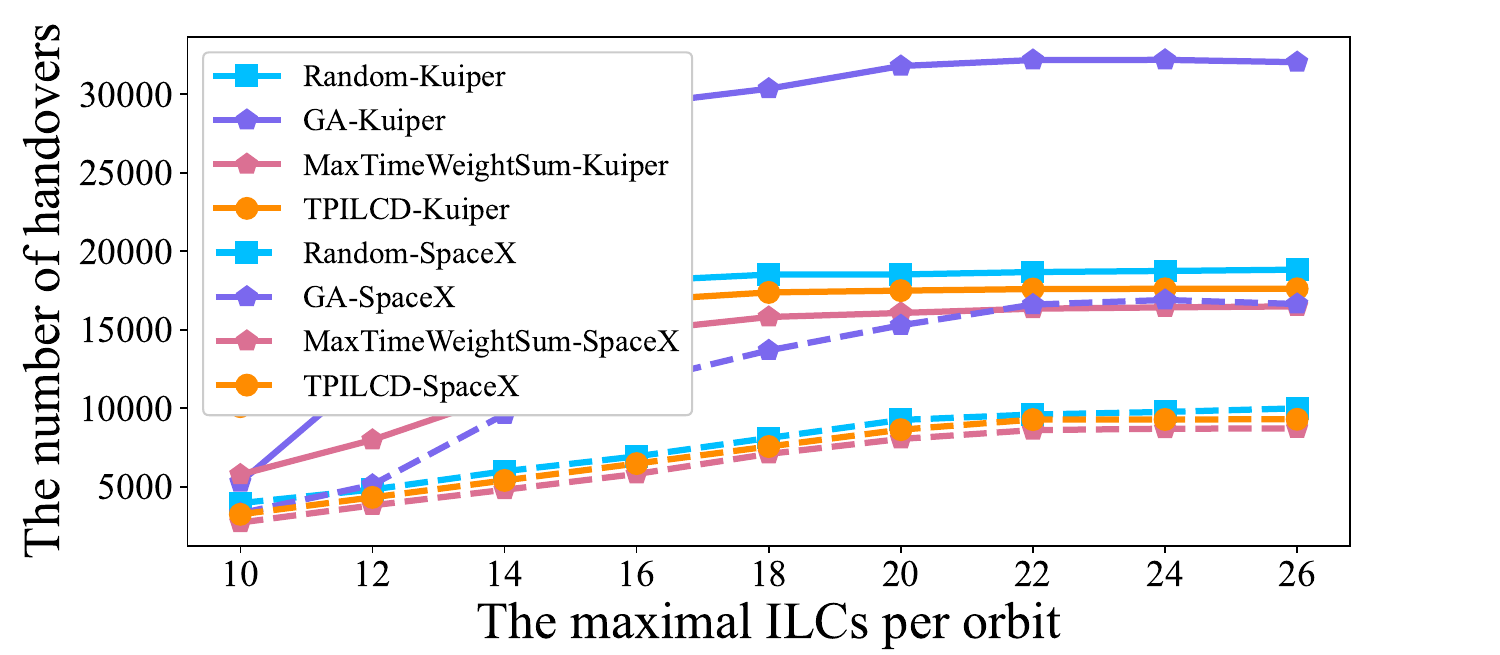}
  \caption{Total handover counts comparison among four algorithms in two mega-constellations.}
  \label{k3}
\end{figure}

%kuiper-fig5:展示层间链的物理拓扑情况：以纬度和轨道索引表示：强调呈纬度对称，每轨均匀分布
% \begin{figure}
%   \centering
%   \subfloat[]{\label{k5}\includegraphics[width=2.7 in]{k_5.eps}} 
%   \\
%   \subfloat[]{\label{k6}\includegraphics[width=1.3 in]{k_6.eps}}
%   \caption{(a) The satellite distribution in the Kuiper-C constellation. (b) The allocation of the common inter-layer connections.}
% \end{figure}
\begin{figure}
  \centering
  \includegraphics[width=2.7 in]{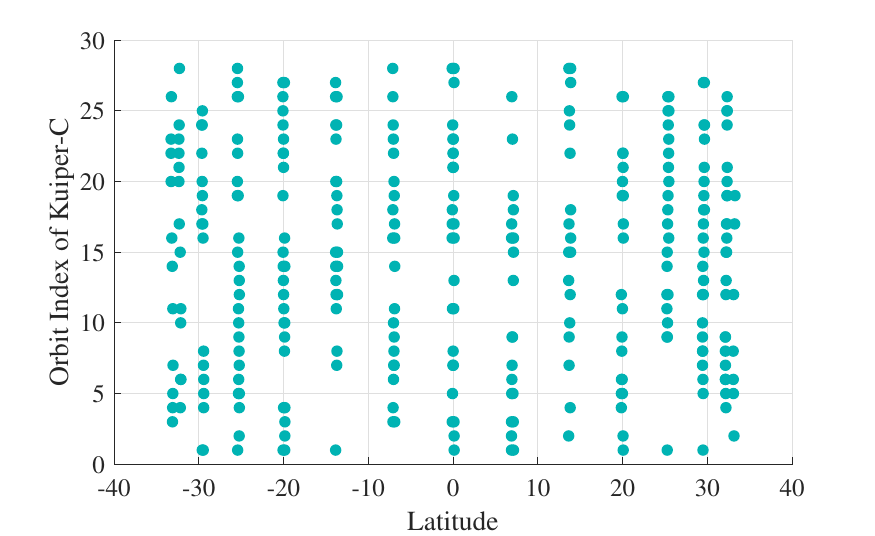}
  \caption{The satellite distribution in the Kuiper-C constellation}
  \label{k5}
  \vspace{-1em}
\end{figure}

\begin{figure}
  \centering
  \includegraphics[width=2.3 in]{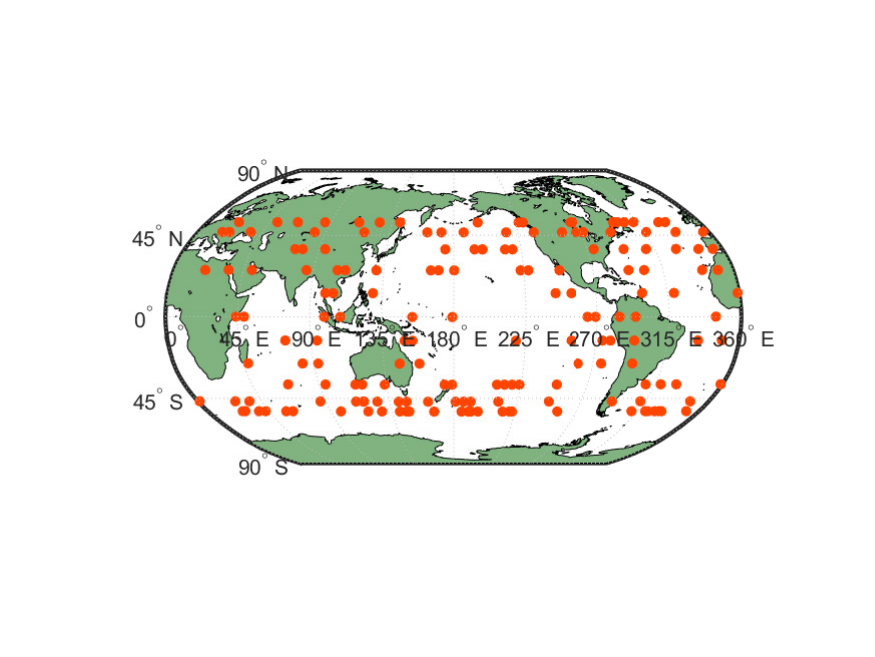}
  \caption{Geographical distribution in the SpaceX-1 constellation}
  \label{k7}
  \vspace{-1em}
\end{figure}

\begin{figure}
  \centering
  \includegraphics[width=2.5 in]{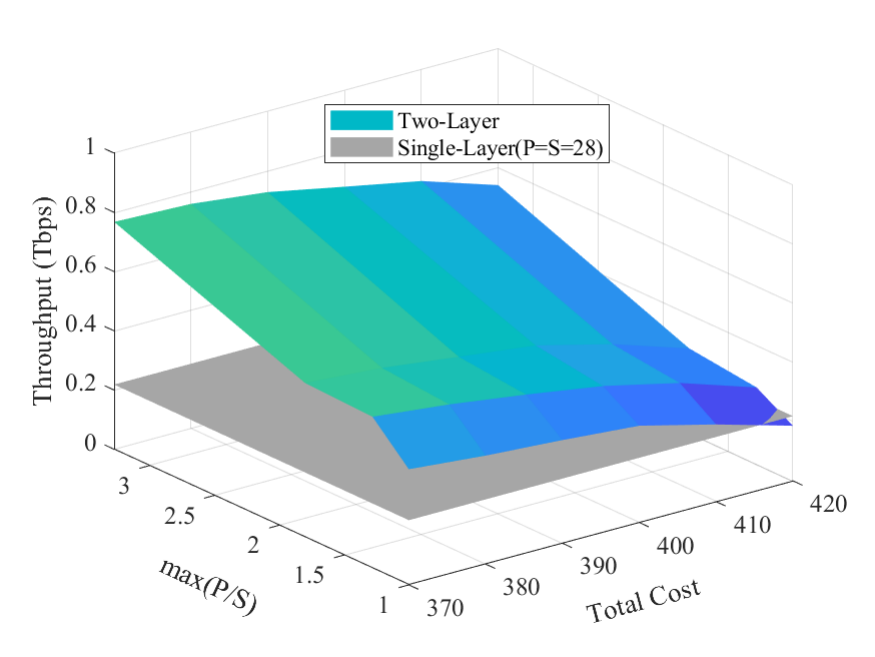}
  \caption{Throughput for different costs and structures.}
  \label{imp1}
  \vspace{-1em}
\end{figure}

%放到仿真里面说明，作为评估吞吐量的一个指标
\subsection{Traffic Evaluation}
Considering that adding ILCs not only can enhance network transmission efficiency, but also can offload the extra burden of monolayer and elevate the service quality, we simulate the real global traffic based on the world internet users \cite{urlglobal} and we refer to the coverage data of the Kuiper constellation in \cite{Kuiper2019}. Substantially, the traffic carried by the MLUDSN can be regarded as multiple traffic flows with different origins and destinations from terrestrial users. Thus we examplify that 1000 cities with a dense population as the origin and destination of traffic flows, and collect 100 thousand city pairs as the flows. Note that we do not utilize the traffic flows as part of modeling the original problem, but as input to evaluate the ILC building strategy. The reason has two, one is that the traffic flows are imprecise and have extreme differences between different time slices;  the other one is that the constraints such as flow conservation, and flow variables render the computational complexity even factorially with satellite and traffic matrix scale. Therefore, we mainly explore the performance indicators to demonstrate the advantages of the MLUDSN and the load balance by cross-layer connections.

%总的系统评估
% \begin{figure}
%   \centering
%   \includegraphics[width=2.7 in]{fig9.eps}
%   \caption{The evaluation system of ILCs.}
%   \label{system}
% \end{figure}

% We define the traffic matrix as $\mathcal{F}$, and the total volume of it is $\left\lvert \mathcal{F} \right\rvert$. Through a certain routing scheme, the link with the highest occupancy ratio can be denoted as $U= \frac{ u_{ij}}{c_{ij}}$, where $u_{ij}$ means the bandwidth occupancy of the link $e_{ij}$, and $c_{ij}$ indicates the capacity of $e_{ij}$. Therefore, we can obtain the maximum throughput of a MLUDSN:
% \begin{equation}
%   Tp = \frac{\left\lvert \mathcal{F} \right\rvert}{U}.
% \end{equation} 
% Actually, different routing schemes can significantly influence the network throughput. in this paper, we mainly consider the congestion control routing scheme to exhibit the advantages of build inter-layer links from the perspective of topology design.
\begin{figure}
  \centering
  \subfloat[]{\label{l1_1}\includegraphics[width=3.4 in]{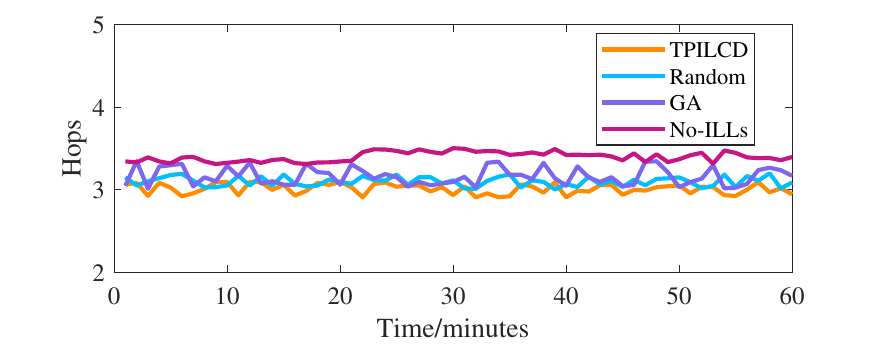}} 
  \\ \vspace{-1em}
  \subfloat[]{\label{l1_2}\includegraphics[width=3.4 in]{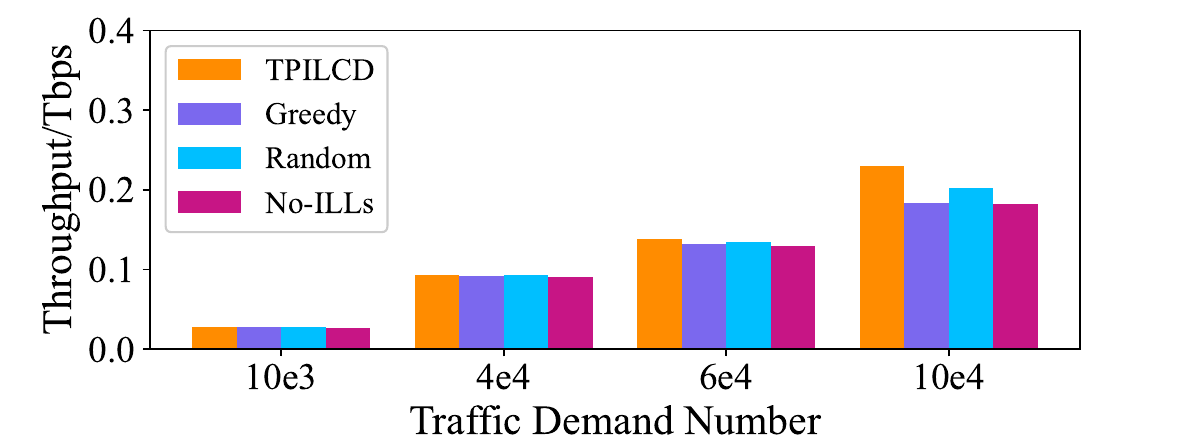}}
  \caption{(a) Fluctuation in hop count over time; (b) Fluctuation in total throughput over time.}
  \vspace{-1em}
\end{figure}

Fig. \ref{l1_1} depicts the change of the average hop count for topologies generated by different algorithms over time. The average hop by the MLUDSN is 3.0094 based on the TPILCD algorithm, and the hops by the No ILCs (No-ILLs) algorithm is 0.4 hops more than that of TPILCD.
With time changes, the hops of these algorithms all fluctuate. The standard deviations of the hops of TPILCD, Random, GA, and No-ILLs are 0.0566, 0.0614, 0.1068, and 0.0578, respectively. The TPILCD algorithm fluctuates the least, which illustrates the pretty load balance between layers. Furthermore, we investigate network throughput of different topologies in Fig. \ref{l1_2}. With the number of traffic flows increases, the total throughput rises of both four schemes. Based on Fig. \ref{t1}, it can be observed that monolayer Kuiper-C can carry about 0.2 Tbit data, and when the traffic flows are over 100 thousand, the Kuiper-C offloads the traffic into the Kuiper-B, so that the throughput obtained by the TPILCD algorithm can enhance over 20.77\% than No-ILLs.

\section{Conclusion}

In this paper, we investigate the impacts of limited ILCs on the MLUDSN's transmission capability. Specifically, to pursue the cost-effective multi-layer constellation, the relationship between the number of ILCs, their deployments and the concomitant variation of network capability is necessary to study. Therefore, we formulate the cost-efficient problem of minimizing network hops by minimizing the number of ILCs. Especially, we reveals that the trend of the APL of the MLUDSN decays logarithmically as ILCs rises, and verify that the number of ILCs should be less than half the scale of single layer in most scenarios. Considering the computational complexity increases exponentially with MLUDSN scale, we propose a TPILCD algorithm to obtain the efficient ILC deployment scheme. In the first phase, we devise an OTLC algorithm and aim to search for the candidate matching sets by narrowing the feasible region and model the APL. To reduce the number of handovers and design a more stable topology, we weigh the time metric to measure the survival duration of each link, and further propose an MTWM algorithm in the second phase to maximize the sum of time weights. Simulation results demonstrate that the ILCs should be deployed symmetrically and evenly on each orbit in the topology and the significant gain of multi-layer collaboration.

\bibliography{ref} % 导入lib，ref为“ref.lib"的文件名
\bibliographystyle{IEEEtran}

% \newpage

 

% % \begin{IEEEbiography}[{\includegraphics[width=1in,height=1.25in,clip,keepaspectratio]{fig1}}]{Michael Shell}
% %\end{IEEEbiography}


% % \begin{IEEEbiographynophoto}{John Doe}
% % Use $\backslash${\tt{begin\{IEEEbiographynophoto\}}} and the author name as the argument followed by the biography text.
% % \end{IEEEbiographynophoto}

\vfill

\end{document}